\documentclass[journal, onecolumn, twoside]{IEEEtran}
\normalsize
\ifCLASSINFOpdf
\else
\fi
\usepackage{setspace}
\usepackage{color}
\usepackage{cite}
\usepackage{pifont}
\usepackage{amsmath}
\usepackage{mathtools}
\usepackage{amsfonts}
\usepackage{amssymb}
\usepackage{amsthm}
\usepackage{tikz}
\usepackage{multirow}
\usepackage{makecell}
\usepackage{mathdots}
\usepackage{arydshln}
\usepackage{cases}
\usepackage{bm}
\usepackage{graphicx}
    \graphicspath{{../}}
    \DeclareGraphicsExtensions{.pdf}
\usepackage[caption=true,font=footnotesize]{subfig}
\usepackage{epstopdf}
\usepackage{multirow}
\usepackage{makecell}
\usepackage{mathdots}
\usepackage{booktabs}
\usepackage{url}
\usepackage{bm}
\usepackage{xtab}
\usepackage{pifont}
\usepackage{array}
\usepackage{algorithm}
\usepackage{algorithmic}
\usepackage{enumerate}
\usetikzlibrary{arrows}
\usepackage{caption}

\allowdisplaybreaks[4]
\usepackage[colorlinks,
            linkcolor=blue,
            anchorcolor=blue,
            citecolor=blue]{hyperref}
\theoremstyle{plain}

\newtheorem{theorem}{Theorem}
\newtheorem{lemma}[theorem]{Lemma}

\newtheorem{definition}[theorem]{Definition}
\newtheorem{corollary}[theorem]{Corollary}
\theoremstyle{definition}
\newtheorem{example}{Example}

\newtheorem{remark}{Remark}

\begin{document}
\title{Duality and Reverse Self-Dual Constructions for Hyperderivative Reed-Solomon Codes}

\author{Hongchang Li,~\IEEEmembership{}
        Zhihao Zhu,~\IEEEmembership{}
        Weijun Fang,~\IEEEmembership{}
        Jun Zhang,~\IEEEmembership{}
     and  Sihuang Hu~\IEEEmembership{}
\IEEEcompsocitemizethanks{\IEEEcompsocthanksitem Hongchang Li, Weijun Fang and Sihuang Hu are with the State Key Laboratory of Cryptography and Digital Economy Security, the Key Laboratory
of Cryptologic Technology and Information Security, Ministry of Education,
and the School of Cyber Science and Technology, Shandong University,
Qingdao 266237, China (emails: lihongchang866@163.com, fwj@sdu.edu.cn, husihuang@sdu.edu.cn).}
\IEEEcompsocitemizethanks{\IEEEcompsocthanksitem Zhihao Zhu and Jun Zhang are with the School of Mathematical Sciences, Capital Normal University, Beijing 100048, China. (emails: 2240501027@cnu.edu.cn, junz@cnu.edu.cn).
}
\thanks{The work is supported in part by the National Key Research and
Development Program of China under Grant No. 2021YFA1001000, the National Natural Science Foundation of China under
Grant Nos. 12441105 and 62571301, and the Shandong Provincial Natural Science Foundation
under Grant Nos. ZR2025QA05 and ZR2026ZD36. (Corresponding Author: Weijun Fang)}
}

\maketitle

\begin{abstract}
Hyperderivative Reed-Solomon (HRS) codes form a class of maximum-distance-separable codes under the Niederreiter-Rosenbloom-Tsfasman metric and may be viewed as a derivative-evaluation extension of classical Reed-Solomon codes. For generalized Reed-Solomon codes, the Euclidean dual is again a generalized Reed-Solomon code. In this paper, we investigate the corresponding duality problem for HRS codes. Using a residue-theoretic argument, we derive an explicit component-wise representation for the Euclidean dual of an HRS code. The formula shows that, in general, the Euclidean dual is not an HRS code. Instead, it is blockwise upper-triangularly equivalent to a reverse-order HRS evaluation code, where the reversal occurs in the hyperderivative orders within each evaluation block. In particular, for full-domain HRS codes with low multiplicity, the
triangular transformations reduce to diagonal scalings, and the Euclidean dual is obtained as the row reversal of an HRS code. Based on this reverse-order dual structure, we further study reverse self-dual HRS codes. We establish explicit criteria for reverse self-duality and construct several families from additive and multiplicative coset structures.
\end{abstract}

 \IEEEpeerreviewmaketitle
\section{Introduction}\label{secI}
 Reed--Solomon (RS) codes are among the most fundamental families in modern algebraic coding theory. Classically defined via the evaluation of polynomials over finite fields, these codes are maximum distance separable (MDS), meaning that they reach the Singleton bound in the Hamming metric and provide optimal error-correction capability for a given redundancy. However, with the evolution of increasingly structured data models and channel models, the traditional Hamming metric is often insufficient to capture the underlying dependencies among code coordinates. This has motivated the study of Niederreiter–Rosenbloom–Tsfasman (NRT) metrics, which provide a refined framework for coding over partially ordered sets.

This metric was first introduced by Niederreiter in 1987 in the context of the uniform distribution of point sets in Euclidean space~\cite{N87, N91, N92}, and was independently applied to coding theory by Rosenbloom and Tsfasman~\cite{RT97}. The associated $P$-distance is a metric on $\mathbb{F}_q^n$, and the NRT distance is the corresponding metric on the matrix space $\mathrm{Mat}_{s \times r}(\mathbb{F}_q)$.  Following the formal introduction of the NRT metric, a fundamental problem is to maximize the
minimum distance for a fixed dimension $k$ and length $n$. To address this, Rosenbloom and Tsfasman established the Singleton bound for the NRT metric, defining the MDS poset codes as those that achieve this theoretical limit~\cite{RT97}. Since then, a wide range of coding-theoretic problems have been studied under the NRT metric, including the fundamental characterization of MDS codes~\cite{Dougherty02, Dougherty04, Hyun08}, the establishment of the MacWilliams duality identity~\cite{Dougherty02b, Kim05, Ozen15}, and the exploration of covering properties~\cite{Castoldi15}. As a primary example of MDS codes under the NRT metric, classical NRT Reed-Solomon codes were introduced by Rosenbloom and Tsfasman~\cite{RT97}. These codes are constructed by evaluating polynomials and their successive derivatives, successfully extending the MDS property to poset structures. In this paper, we study the class of Hyperderivative Reed-Solomon (HRS) codes, which serve as a generalization of the classical NRT Reed-Solomon codes.

Informally, an HRS code is obtained by evaluating a polynomial together 
with its first several hyperderivatives, also called Hasse derivatives, at 
a set of prescribed field elements, possibly equipped with coordinate-wise 
multipliers. More precisely, let $\mathbb F_q$ be the finite field with 
$q$ elements, and let $r$, $s$, and $k$ be positive integers satisfying 
$r\le q$ and $k\le rs$. Let  
$\alpha_1,\ldots,\alpha_r\in \mathbb F_q$ be pairwise distinct evaluation 
points, and let $v_{i,j}\in \mathbb F_q^\ast$ be nonzero multipliers. 
Then the codewords of the corresponding HRS code are of the form
\[
\bigl(v_{i,j}\partial^{(i-1)}f(\alpha_j)\bigr)_
{1\le i\le s,\ 1\le j\le r},
\]
where $f$ ranges over polynomials of degree at most $k-1$, and 
$\partial^{(i)}$ denotes the $i$-th hyperderivative. The formal 
definition of hyperderivatives and the precise definition of HRS codes 
are recalled in Section~\ref{SecII}. 
When $s=1$, an HRS code reduces to a  generalized 
Reed-Solomon (GRS) code. In recent years, various structural and algorithmic 
aspects of linear codes under the NRT metric have been investigated~\cite{Ozen06, Panek10, Wang25}. For HRS codes in particular, Gu et al. 
studied the dimensions of their Schur squares~\cite{Gu26a} and developed a Welch-Berlekamp type algorithm for unique decoding~\cite{Gu26b}. 
Nevertheless, some fundamental algebraic properties of HRS codes remain 
less understood, especially their dual structure.

For classical GRS codes, it is well known that the Euclidean dual of a GRS code is again a GRS code. This is one of the basic structural properties of GRS codes and has many important applications, such as explicit descriptions of parity-check matrices, syndrome-based decoding, and coding schemes for distributed storage. It is therefore natural to ask whether an analogous duality property holds for HRS codes. Recently, Can and Horowitz~\cite{CH26} initiated the study of duality for HRS codes and claimed that the Euclidean dual of an HRS code is again an HRS code. However, this claim is not correct as stated.

In this paper, we investigate the duality and reverse self-duality of HRS codes. Our first contribution is an explicit component-wise representation for the Euclidean dual of an HRS code (see Theorem \ref{duality-thm}). This formula shows that the dual has a reverse-order hyperderivative evaluation structure, where the Hasse derivative orders are reversed within each evaluation block. It also gives the precise correction to the closure assertion in~\cite{CH26}. In general, the Euclidean dual of an HRS code is not an HRS code in the strict sense of Definition~\ref{HRS-def}. Rather, it is blockwise upper-triangularly equivalent to a reverse-order HRS evaluation code, with the local triangular transformations determined by the expansions of certain rational functions at the evaluation points (see Remark \ref{rem1}). In an important full-domain low-multiplicity case, these triangular transformations reduce to diagonal scalings, and the row-reversed dual is again an HRS code (see Corollary \ref{Full-Evaluation-Points}).

Motivated by this dual structure, we then study reverse self-dual HRS codes. We first establish a moment criterion for reverse self-duality in
the multiplicity-two case (see Theorem \ref{thm:criterion}) and construct several families from additive
and multiplicative cosets (see Theorems \ref{thm:additive-single}, \ref{thm:mult-single} and \ref{thm:mult-s=2}). We further extend the criterion to general
multiplicity under a natural symmetry condition on the multipliers (see Theorem \ref{thm:general_criterion}), and
derive corresponding constructions using additive and multiplicative
structures of finite fields (see Theorems \ref{thm:additive-p-power}, \ref{thm:affine-additive-construction} and \ref{mul--pullback-general}).

The remainder of this paper is organized as follows. Section~\ref{SecII}   recalls
the necessary preliminaries on the NRT metric, Hasse derivatives, HRS
codes, and residues. Section~\ref{secIII} derives the explicit Euclidean dual form
of HRS codes and explains its blockwise triangular reverse-order
structure. Section~\ref{secIV} studies reverse self-dual HRS codes, including the
multiplicity-two case and its extension to general multiplicity.
Section~\ref{secV} concludes the paper.

\section{Preliminaries }\label{SecII}
In this section, we recall several basic notions that will be used 
throughout the paper, including the NRT metric, Hasse derivatives, 
Hyperderivative Reed-Solomon codes and the residue theorem.

\subsection{NRT Metric}

Let $P = ([n], \le)$ be a partially ordered set (poset) on the set $[n]:= \{1, \dots, n\}$. For any $v \in \mathbb{F}_q^n$, let ${\rm supp}(v) = \{i : v_i \neq 0\}$. The $P$-weight of $v$ is defined as the cardinality of the lower-order ideal generated by ${\rm supp}(v)$:
\[  \omega_P(v) = |\{j \in [n] : \exists~i \in {\rm supp}(v) \text{ such that } j \le i\}|.
\]
  For $v, w \in \mathbb{F}_q^n$, the $P$-metric of $v$ and $w$ is defined by 
\[
d_P(v, w) =\omega_P(v - w).
\]
The NRT metric (Niederreiter-Rosenbloom-Tsfasman) is a specific poset metric where the poset $C(s,r)$ is a disjoint union of $r$ chains, each of length $s$.  Let $ V := (a_{ij})_{\substack{i=1,\dots,s \\ j=1,\dots,r}} \in \operatorname{Mat}_{s \times r}(\mathbb{F}_q) $.
Let $ V_1, \dots, V_r $ denote the columns of $ V $. The NRT weight of the column $ V_j $, denoted $ \omega_{C(s,1)}(V_j) $, is defined as $$\omega_{C(s,1)}(V_j)=\begin{cases}
    s-\iota(V_{j})+1 & \text{if}\quad V_{j} \neq \mathbf{0}; \\
    0  & \text{if}  \quad V_{j}=\mathbf{0}.
\end{cases}$$
where $\iota(V_{j}) := \min\{i\in[s]\mid a_{ij}\neq 0\}.$ For a matrix $V \in \operatorname{Mat}_{s \times r}(\mathbb{F}_q)$, the NRT weight is given by
\[  \omega_{\mathrm{NRT}}(V) = \sum_{j=1}^r \omega_{C(s,1)}(V_j).\]
The resulting metric will be denoted by $ d_{C(s,r)} $.

\begin{example}
	To illustrate the NRT weight calculation and its connection to the chain poset, let $s$ be a fixed positive integer. Consider a vector $\mathbf{a} = (a_1, a_2, \dots, a_s) \in \mathbb{F}_q^s$. We represent $\mathbf{a}$ by reversing the order of the entries and writing the result as a column vector:
	\[
	\mathbf{a}^\top_{rev} = [a_s, a_{s-1}, \dots, a_1]^\top.
	\]
	The structure of the NRT poset $\mathcal{C}(4, 1)$ for the case $s=4$, which forms a single chain, is presented in Figure \ref{fig:hasse_c41}.
	
	\begin{figure}[htbp]
		\centering
		\begin{tikzpicture}[node distance=1.2cm, every node/.style={circle, draw, inner sep=2pt, font=\small}]
			% Nodes
			\node (n4) {4};
			\node (n3) [below of=n4] {3};
			\node (n2) [below of=n3] {2};
			\node (n1) [below of=n2] {1};
			
			% Edges
			\draw (n4) -- (n3);
			\draw (n3) -- (n2);
			\draw (n2) -- (n1);
		\end{tikzpicture}
		\caption{The Hasse diagram of $\mathcal{C}(4,1)$}
		\label{fig:hasse_c41}
	\end{figure}
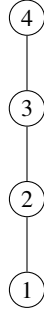
	
	More specifically, let $s=4$ and consider a vector $\mathbf{v} = (0, v_2, v_3, 0) \in \mathbb{F}_q^4$ where $v_2, v_3 \neq 0$. Its representation as a column matrix is the following
	\[
	V = \begin{bmatrix} 0 \\ v_3 \\ v_2 \\ 0 \end{bmatrix} \begin{array}{l} \leftarrow \text{poset element } 4 \\ \leftarrow \text{poset element } 3 \\ \leftarrow \text{poset element } 2 \\ \leftarrow \text{poset element } 1 \end{array}
	\]
	The NRT weight $\omega_{\mathcal{C}(4,1)}(V)$ is determined by the highest index $i$ such that the $i$-th component is non-zero. In this case, the non-zero entries are at positions 2 and 3. The maximal non-zero element in the poset is $3$. Therefore, 
	\[
	\omega_{\mathcal{C}(4,1)}(V) = 3.
	\]
	This confirms the formula $s - i + 1 = 4 - 2 + 1 = 3$, where $i$ is the index of the first non-zero row from the top.
\end{example}

\begin{theorem}\cite{Hyun08}
	Let $\mathcal{C} \subseteq \operatorname{Mat}_{s \times r}(\mathbb{F}_q)$ be an NRT code. The minimum NRT distance of $\mathcal{C}$, denoted by $d_{\mathrm{NRT}}(\mathcal{C})$, is defined as the minimum nonzero NRT weight $w_{\mathrm{NRT}}(V - W)$ for any distinct $V, W \in \mathcal{C}$. Then, the cardinality of $\mathcal{C}$ satisfies the following
	\begin{equation*}
		|\mathcal{C}| \leq q^{rs - d_{\mathrm{NRT}}(\mathcal{C}) + 1}.
	\end{equation*}
	In particular, if $\mathcal{C}$ is a linear $[rs, k]$ NRT code, the minimum distance satisfies the following
	\begin{equation*}
		d_{\mathrm{NRT}}(\mathcal{C}) \leq rs - k + 1.
	\end{equation*}
\end{theorem}

\begin{definition}
	A code $\mathcal{C} \subseteq \operatorname{Mat}_{s \times r}(\mathbb{F}_q)$ that employs the NRT-metric is referred to as an NRT code. A linear NRT code with parameters $[rs, k, d_ {\mathrm{NRT}}]$ is called an MDS NRT-metric code if it achieves the equality $d_{\mathrm{NRT}} = rs - k + 1$.
\end{definition}

\subsection{Hyperderivatives and HRS Codes}

We first recall the notion of hyperderivatives, also called 
Hasse derivatives, and then give the formal definition of 
Hyperderivative Reed-Solomon codes.

\begin{definition}
Let $f(x)=\sum_{\ell=0}^{n} f_\ell x^\ell \in \mathbb{F}_q[x]$.
For a nonnegative integer $j$, the $j$-th hyperderivative, or Hasse
derivative, of $f$ is defined by
\[
\partial^{(j)}f(x)
=
\sum_{\ell=0}^{n}\binom{\ell}{j}f_\ell x^{\ell-j},
\]
where \[\binom{\ell}{j}:=\left\{\begin{array}{ll}\frac{\ell!}{(\ell-j)!j!},& \text { if } 0 \leq j \leq \ell,\\0, & \text { otherwise. }\end{array}\right.\]
Here the coefficient $\binom{\ell}{j}$ is understood as the reduction modulo
$p$ of the corresponding integer binomial coefficient. In particular, $\partial^{(0)}f=f$.
\end{definition}

Equivalently, Hasse derivatives are characterized by the formal Taylor
expansion
\[
f(x+z)=\sum_{j\ge 0}\partial^{(j)}f(x)z^j.
\]
Thus $\partial^{(j)}f(x)$ is the coefficient of $z^j$ in $f(x+z)$.
For any $u\in \mathbb{F}_q$, one also has
\[
f(x)=\sum_{j\ge 0}\partial^{(j)}f(u)(x-u)^j.
\]
We use the same notation for the local expansion of a rational function that is regular at the expansion point.

The following lemma presents the product rule for Hasse derivatives.

\begin{lemma}[\cite{Dvir13}]\label{product-rule}
Let $\mathbb{F}$ be a field and let $g_1,\ldots,g_\ell\in \mathbb{F}[x]$. Then
\[
\partial^{(n)}(g_1\cdots g_\ell)
=
\sum_{\substack{n_1,\ldots,n_\ell\ge 0\\
n_1+\cdots+n_\ell=n}}
\bigl(\partial^{(n_1)}g_1\bigr)\cdots
\bigl(\partial^{(n_\ell)}g_\ell\bigr).
\]
\end{lemma}

For $c\in \mathbb{F}$ and $t\ge 0$, we have
\[\partial^{(n)}(x-c)^{t}=\left\{\begin{array}{ll}\binom{t}{n}(x-c)^{t-n}, & 0 \leq n \leq t, \\0, & n>t .\end{array}\right.\]
In particular,
\[
\partial^{(t)}(x-c)^t=1.
\]
Hasse derivatives record multiplicities in the following sense. A point
$\alpha$ is a root of $f$ of multiplicity $m$ if and only if
\[
\partial^{(j)}f(\alpha)=0,\quad 0\le j\le m-1, \text{ and }
\partial^{(m)}f(\alpha)\ne 0.
\]
We denote by $\nu_f(\alpha)$ the order of vanishing of $f$ at
$\alpha$.

We now define Hyperderivative Reed--Solomon codes. Let $\mathcal{A}=\{\alpha_1,\ldots,\alpha_r \}$ be an evaluation subset of
$r$ pairwise distinct elements of $\mathbb F_q$, and let
\[
 V=(v_{i,j})_{1\le i\le s,\ 1\le j\le r}\in(\mathbb F_q^\ast)^{s\times r}
\]
be a multiplier matrix. The integer $s$ will be referred to as the
multiplicity parameter, since at each evaluation point $\alpha_j$ the          
code evaluates the polynomial together with its first $s-1$ Hasse
derivatives. For an integer $h\ge 0$, define
\[
\mathcal P_{<h}=\{f\in\mathbb F_q[x]: \deg f<h\},
\]
with the convention that $\mathcal P_{<0}=\{0\}.$
\begin{definition}[Hyperderivative Reed--Solomon code]\label{HRS-def}
For $0 \le k \le rs$, the Hyperderivative Reed--Solomon code with dimension $k$,
evaluation set $\mathcal{A}$, multiplier matrix $V$, multiplicity $s$,
and length  $sr$, denoted by
\[
\operatorname{HRS}_k(\mathcal{A},V,s,r),
\]
is the image of the evaluation map
\[
\operatorname{Ev}_{\mathcal{A},V}:\mathcal{P}_{<k}
\longrightarrow
\operatorname{Mat}_{s\times r}(\mathbb{F}_q)
\]
defined by
\[
f\longmapsto
\begin{bmatrix}
v_{1,1}\partial^{(0)}f(\alpha_1) &
v_{1,2}\partial^{(0)}f(\alpha_2) &
\cdots &
v_{1,r}\partial^{(0)}f(\alpha_r)
\\
v_{2,1}\partial^{(1)}f(\alpha_1) &
v_{2,2}\partial^{(1)}f(\alpha_2) &
\cdots &
v_{2,r}\partial^{(1)}f(\alpha_r)
\\
\vdots & \vdots &  & \vdots
\\
v_{s,1}\partial^{(s-1)}f(\alpha_1) &
v_{s,2}\partial^{(s-1)}f(\alpha_2) &
\cdots &
v_{s,r}\partial^{(s-1)}f(\alpha_r)
\end{bmatrix}.
\]
Equivalently,
\[
\operatorname{HRS}_k(\mathcal{A},V,s,r)
=
\left\{
\begin{bmatrix}
v_{1,1}\partial^{(0)}f(\alpha_1) &
v_{1,2}\partial^{(0)}f(\alpha_2) &
\cdots &
v_{1,r}\partial^{(0)}f(\alpha_r)
\\
v_{2,1}\partial^{(1)}f(\alpha_1) &
v_{2,2}\partial^{(1)}f(\alpha_2) &
\cdots &
v_{2,r}\partial^{(1)}f(\alpha_r)
\\
\vdots & \vdots & & \vdots
\\
v_{s,1}\partial^{(s-1)}f(\alpha_1) &
v_{s,2}\partial^{(s-1)}f(\alpha_2) &
\cdots &
v_{s,r}\partial^{(s-1)}f(\alpha_r)
\end{bmatrix}
:
f\in\mathcal{P}_{<k}
\right\}.
\]
\end{definition}
When $s=1$, the above definition reduces to the classical generalized
Reed-Solomon code
\[\operatorname{HRS}_k(\mathcal{A},V,1,r)
=
\left\{
\bigl(v_{1,j}f(\alpha_j)\bigr)_{1\le j\le r}
:
f\in\mathcal{P}_{<k}
\right\}.\]
Thus, HRS codes may be viewed as derivative-evaluation generalizations of
GRS codes. When $V$ is the all-ones matrix, HRS codes reduce to the known NRT RS codes \cite{RT97, Skriganov01}.

The MDS property of HRS codes with respect to the
NRT metric is entirely analogous to that of NRT RS codes, as established in \cite{Skriganov01}. For completeness, we include the proof.

\begin{theorem}
For $1\le k\le rs$, $\operatorname{HRS}_k(\mathcal{A},V,s,r)$
is an $[rs,k]$ MDS code with respect to the NRT metric, equivalently,
\[
d_{\mathrm{NRT}}
\bigl(\operatorname{HRS}_k(\mathcal{A},V,s,r)\bigr)
=
rs-k+1.
\]
\end{theorem}

\begin{proof}
The evaluation map is injective. Indeed, if a polynomial
$f\in\mathcal{P}_{<k}$ is mapped to the zero codeword, then
\[
\partial^{(i)}f(\alpha_j)=0,
\qquad
0\le i\le s-1,\ 1\le j\le r.
\]
Hence, each $\alpha_j$ is a root of $f$ of multiplicity at least $s$.
Thus, $f$ has at least $rs$ roots counted with multiplicity. Since
$\deg f<k\le rs$, we must have $f\equiv 0$. Therefore, the dimension of the
code is $k$.

Let $f \in \mathcal{P}_{<k}$ be a nonzero polynomial, and set $\nu_j:=\nu_f(\alpha_j)$ for $1\le j\le r.$
In the $j$-th column of the corresponding codeword, the first nonzero
entry occurs in the row $\nu_j+1$ if $\nu_j<s$, while the entire column is zero if $\nu_j\ge s$. Hence, the NRT weight of this column is $s-\min\{\nu_j,s\}.$
Therefore
\[
\omega_{\mathrm{NRT}}\bigl(\operatorname{Ev}_{\mathcal{A},V}(f)\bigr)
=
\sum_{j=1}^{r}\bigl(s-\min\{\nu_j,s\}\bigr)
=
rs-\sum_{j=1}^{r}\min\{\nu_j,s\}.
\]
Since $\sum_{j=1}^{r}\min\{\nu_j,s\}
\le
\sum_{j=1}^{r}\nu_j
\le
\deg f
\le
k-1$,
we obtain $\omega_{\mathrm{NRT}}\bigl(\operatorname{Ev}_{\mathcal{A},V}(f)\bigr)
\ge
rs-k+1$. Thus $d_{\mathrm{NRT}}\ge rs-k+1$.

On the other hand, by the Singleton bound (see Corollary 2.2 in \cite{Hyun08}), $d_{\mathrm{NRT}}\leq  rs-k+1$.
Hence $d_{\rm NRT}\bigl(\operatorname{HRS}_k(\mathcal{A},V,s,r)\bigr)=rs-k+1$ and  the HRS code is MDS.
\end{proof}

\subsection{Residues}
The explicit duality formula for HRS codes is derived through a
residue-theoretic argument on the projective line. We first recall the residue facts needed for this argument in a form adapted to our finite-field setting. 

Throughout this subsection, for a rational function $R(x)\in \mathbb F(x)$
and a point $P\in \mathbb P^1$, the notation
$\operatorname{Res}_{P}R(x)$ means the residue of the rational differential
$R(x)\,dx$ at $P$. In particular, if $a\in \mathbb F$ and
\[
R(x)=\sum_{\ell\ge \ell_0} c_\ell (x-a)^\ell
\]
is the Laurent expansion of $R(x)$ at $a$, then
\[
\operatorname{Res}_{x=a}R(x):=c_{-1}.
\]
This agrees with the usual residue of $R(x)\,dx$ at the rational point
$x=a$. At the point at infinity, the residue is computed by using the local
parameter $t=1/x$.

We shall repeatedly use the following elementary consequence of the residue
theorem on $\mathbb P^1$.

\begin{theorem}[Residue theorem on $\mathbb P^1$ {\cite{Stichtenoth09}}]\label{resthm}
Let $\mathbb F$ be a finite field, and let $R(x)\in \mathbb F(x)$.
Assume that all finite poles of $R(x)$ are $\mathbb F$-rational points
$\alpha_1,\ldots,\alpha_m$. Then
\[
\sum_{\ell=1}^{m}\operatorname{Res}_{x=\alpha_\ell} R(x)
+
\operatorname{Res}_{x=\infty} R(x)
=0.
\]
In particular, if $R(x)=A(x)/B(x)$, where $A(x),B(x)\in\mathbb F[x]$ and
\[
\deg A\le \deg B-2,
\]
then $\operatorname{Res}_{x=\infty}R(x)=0$, and hence
\[
\sum_{\ell=1}^{m}\operatorname{Res}_{x=\alpha_\ell} R(x)=0 .
\]
\end{theorem}

We also record the local computation that will be used later. If
$\phi(x)\in \mathbb F(x)$ is regular at $a\in\mathbb F$, then it has the
Hasse--Taylor expansion
\[
\phi(x)=\sum_{\ell\ge0}\partial^{(\ell)}\phi(a)(x-a)^\ell .
\]
Therefore, for every positive integer $s$,
\[
\operatorname{Res}_{x=a}\frac{\phi(x)}{(x-a)^s}
=
\partial^{(s-1)}\phi(a).
\]

In Section~\ref{secIII}, we apply the theorem to rational functions of the form
\[
R(x)=
\frac{f(x)g(x)}
{\prod_{\ell=1}^{r}(x-\alpha_\ell)^s},
\qquad \alpha_1,\ldots,\alpha_r\in\mathbb F_q .
\]
All finite poles are then $\mathbb F_q$-rational, and the local formula
above expresses each finite residue in terms of Hasse derivatives at the
corresponding evaluation point.
    \section{The dual codes of HRS codes}\label{secIII}
In this section, we determine the Euclidean dual of an HRS code. For any $A=(a_{i,j}),\ B=(b_{i,j})\in \operatorname{Mat}_{s\times r}(\mathbb{F}_q)$, we define their Euclidean inner product as
\[
\langle A, B\rangle
=
\sum_{i=1}^{s}\sum_{j=1}^{r}a_{i,j}b_{i,j}.
\]
For a linear code
$C\subseteq \operatorname{Mat}_{s\times r}(\mathbb{F}_q)$, its Euclidean dual is given as
\[
C^\perp
=
\left\{
B\in \operatorname{Mat}_{s\times r}(\mathbb{F}_q):
\langle A,B\rangle=0\ \text{for all } A\in C
\right\}.
\]
Let $\mathcal{A}=\{\alpha_1,\ldots,\alpha_r\}$
be a subset of $r$ distinct elements of $\mathbb{F}_q$, and
$V=(v_{i,j})\in(\mathbb F_q^\ast)^{s\times r}$.
Define
\[
A(x):=\prod_{\ell=1}^{r}(x-\alpha_\ell)^s.
\]
For each $j\in\{1,\ldots,r\}$, put
\[
A_j(x):=\prod_{\ell\ne j}(x-\alpha_\ell)^s.
\]
Then
\[
A(x)=(x-\alpha_j)^s A_j(x),
\qquad
A_j(\alpha_j)\ne 0.
\]

The following theorem provides an explicit description of the Euclidean dual
of an HRS code.
\begin{theorem}\label{duality-thm}
Let
$C=\operatorname{HRS}_k(\mathcal{A},V,s,r)$ be an HRS code. Set $N=rs$. For every polynomial $g\in \mathcal{P}_{<N-k}$, define a matrix
$W(g)=(w_{i,j}(g))\in \operatorname{Mat}_{s\times r}(\mathbb{F}_q)$ by
\[
w_{i,j}(g)
=
\frac{1}{v_{i,j}}
\partial^{(s-i)}
\left(
\frac{g(x)}{A_j(x)}
\right)\bigg|_{x=\alpha_j},
\qquad
1\le i\le s,\ 1\le j\le r.
\]
Then the dual code is given by
\[
C^\perp
=
\{W(g): g\in \mathcal{P}_{<N-k}\}.
\]
\end{theorem}
\begin{proof}
If $N-k=0$, then $\mathcal P_{<N-k}=\{0\}$, and the assertion is immediate.
We assume henceforth that $N-k>0$. For any $f\in\mathcal{P}_{<k}$, the corresponding codeword of $C$ is
\[
\operatorname{Ev}_{\mathcal{A},V}(f)
=
\bigl(v_{i,j}\partial^{(i-1)}f(\alpha_j)\bigr)_
{1\le i\le s,\ 1\le j\le r}.
\]
We show that this codeword is orthogonal to $W(g)$, for any $g\in\mathcal{P}_{<N-k}$.

For a fixed $j$, denote the contribution of the $j$-th column to the inner product by $S_j$. Substituting the definition of $w_{i,j}(g)$ and applying the product rule for Hasse derivatives,  we obtain $$\begin{aligned}
S_j &= \sum_{i=1}^{s} v_{i,j}\partial^{(i-1)}f(\alpha_j)\, w_{i,j}(g) \\
&= \sum_{i=1}^{s} \partial^{(i-1)}f(\alpha_j) \partial^{(s-i)} \left( \frac{g(x)}{A_j(x)} \right)\bigg|_{x=\alpha_j} \\
&= \partial^{(s-1)} \left( \frac{f(x)g(x)}{A_j(x)} \right)\bigg|_{x=\alpha_j}.
\end{aligned}$$

Since $A(x)=(x-\alpha_j)^sA_j(x),$ the rational function
\[
R(x):=\frac{f(x)g(x)}{A(x)}=
\frac{1}{(x-\alpha_j)^s}
\cdot
\frac{f(x)g(x)}{A_j(x)}.
\]
Thus, the residue of $R(x)$ at $x=\alpha_j$ is
\[
\operatorname{Res}_{x=\alpha_j}R(x)
=
\partial^{(s-1)}
\left(
\frac{f(x)g(x)}{A_j(x)}
\right)\bigg|_{x=\alpha_j}
=
S_j.
\]

Therefore 
\[
\left\langle \operatorname{Ev}_{\mathcal{A},V}(f), W(g)\right\rangle
=
\sum_{j=1}^{r}S_j
=
\sum_{j=1}^{r}
\operatorname{Res}_{x=\alpha_j}
\frac{f(x)g(x)}{A(x)}.
\]
Since $\deg(fg) \le (k-1)+(N-k-1) = N-2$ and $\deg A=N$, the rational function $R(x)=f(x)g(x)/A(x)$ has zero residue at infinity. By the residue theorem (Theorem \ref{resthm}),
\[
\sum_{j=1}^{r}
\operatorname{Res}_{x=\alpha_j}R(x)
=
0.
\]
Hence, for all $f\in\mathcal{P}_{<k}$ and  $g\in\mathcal{P}_{<N-k}$,
\[
\left\langle \operatorname{Ev}_{\mathcal{A},V}(f), W(g)\right\rangle=0.
\]
Thus
\[
\{W(g):g\in\mathcal{P}_{<N-k}\}
\subseteq C^\perp.
\]

It remains to compare dimensions. We first show that the map
\[
g\longmapsto W(g)
\]
from $\mathcal{P}_{<N-k}$ to
$\operatorname{Mat}_{s\times r}(\mathbb{F}_q)$ is injective. Suppose
$W(g)=0$. Then for every $j$ and every $0\le a\le s-1$,
\[
\partial^{(a)}
\left(
\frac{g(x)}{A_j(x)}
\right)\bigg|_{x=\alpha_j}
=
0.
\]
Since $A_j(\alpha_j)\ne 0$,  $g$ must have a zero of
multiplicity at least $s$ at each $\alpha_j$. Hence $A(x)\mid g(x)$.
But $\deg A=N$, whereas
\[
\deg g<N-k\le N-1.
\]
Therefore $g=0$, proving injectivity.

Consequently,
\[
\dim\{W(g):g\in\mathcal{P}_{<N-k}\}=N-k.
\]
Since $C$ has dimension $k$, we have $\dim C^\perp=N-k.$
Thus, the above inclusion is an equality. This completes the proof.
\end{proof}

\begin{remark}\label{rem1}
Theorem~\ref{duality-thm}  shows two features of the Euclidean dual of HRS codes.

First, the Hasse derivative orders are reversed within each evaluation block. Indeed, the $i$-th row of a dual codeword involves derivatives
of order $s-i$, while the $i$-th row of the original HRS code involves derivatives of order $i-1$.

Second, in general, there is an additional local triangular change of
coordinates. To see this, fix an evaluation point $\alpha_j$ and define
the reverse-order derivative vector
\[
\mathbf d_j(g)
=
\begin{pmatrix}
\partial^{(s-1)}g(\alpha_j)\\
\partial^{(s-2)}g(\alpha_j)\\
\vdots\\
\partial^{(0)}g(\alpha_j)
\end{pmatrix}
\]
and the $j$-th column of the dual codeword
\[
\mathbf w_j(g)
=
\begin{pmatrix}
w_{1,j}(g)\\
w_{2,j}(g)\\
\vdots\\
w_{s,j}(g)
\end{pmatrix}.
\]
For any $0 \leq i \leq s-1$, let $\gamma_{i,j}=\partial^{(i)}
\left(
\frac{1}{A_j(x)}
\right)\bigg|_{x=\alpha_j}$. Then the component formula in Theorem~\ref{duality-thm} can be written as
\[
\mathbf w_j(g)=D_jU_j\mathbf d_j(g),
\]
where
\[
D_j=
\operatorname{diag}
\left(
v_{1,j}^{-1},v_{2,j}^{-1},\ldots,v_{s,j}^{-1}
\right),
\]
and
\[
U_j=
\begin{pmatrix}
\gamma_{0,j} & \gamma_{1,j} & \gamma_{2,j} & \cdots & \gamma_{s-1,j}\\
0 & \gamma_{0,j} & \gamma_{1,j} & \cdots & \gamma_{s-2,j}\\
0 & 0 & \gamma_{0,j} & \cdots & \gamma_{s-3,j}\\
\vdots & \vdots & \vdots & \ddots & \vdots\\
0 & 0 & 0 & \cdots & \gamma_{0,j}
\end{pmatrix}.
\]
Here $U_j$ is upper triangular with respect to the reverse derivative
order. Since
\[
\gamma_{0,j}
=
\frac{1}{A_j(\alpha_j)}
\ne 0,
\]
the matrix $U_j$ is invertible. Thus, for each evaluation block, the
dual coordinates are obtained from the reverse-order Hasse derivative
vector by an invertible upper-triangular transformation, followed by a
diagonal scaling.

Consequently, in the general case, $C^\perp$ is not necessarily an HRS
code. Rather, it is blockwise
upper-triangularly equivalent to a reverse-order HRS evaluation code. The
transformations are local in the evaluation points, since the matrices
$U_j$ may depend on $j$.  
\end{remark}

\begin{remark}
Theorem~\ref{duality-thm} provides an explicit formula for the Euclidean dual of an HRS code. When the multiplicity $s=1$, the HRS code reduces to a classical GRS code, and the formula recovers the well-known GRS duality:
\[
w_{1,j}(g) = \frac{1}{v_{1,j}} \frac{g(\alpha_j)}{\prod_{\ell\neq j} (\alpha_j - \alpha_\ell)}.
\]
For general $s>1$, the Euclidean dual is generally not an HRS code with the same evaluation order and multiplicities. Instead, it can be obtained from a reverse-order HRS evaluation code through invertible local upper-triangular transformations in each block. Consequently, the GRS-type closure property previously asserted in~\cite{CH26}, namely that the Euclidean dual of an HRS code is again an HRS code, is not correct. The next example computes the dual code explicitly in a concrete case and exhibits a counterexample to the claimed closure property.
\end{remark}

\begin{example}\label{example-dual}
Consider an HRS code over $\mathbb F_5$ with multiplicity $s=2$, dimension $k=4$, and evaluation set
$\mathcal A=\{0,1,2\}\subseteq \mathbb F_5$. Let $V$ be the all-one
multiplier matrix. Then $N=rs=6$, and $C=\operatorname{HRS}_4(\mathcal A,V,2,3)$ is generated by the evaluations of $\{1,x,x^2,x^3\}$ together with their
first Hasse derivatives. The corresponding codewords are
\[
M(1)=
\begin{bmatrix}
1&1&1\\
0&0&0
\end{bmatrix},
\qquad
M(x)=
\begin{bmatrix}
0&1&2\\
1&1&1
\end{bmatrix},
\]
and
\[
M(x^2)=
\begin{bmatrix}
0&1&4\\
0&2&4
\end{bmatrix},
\qquad
M(x^3)=
\begin{bmatrix}
0&1&3\\
0&3&2
\end{bmatrix}.
\]

We now compute the dual code using Theorem~\ref{duality-thm}. Since
$N-k=2$, the dual code is obtained from polynomials $g\in\mathcal P_{<2}$.
For $j=1,2,3$, set $A_j(x)=\prod_{m\ne j}(x-\alpha_m)^2$.
A direct calculation gives $\bigl(A_1(0),\partial^{(1)}A_1(0)\bigr)=(4,3),~
\bigl(A_2(1),\partial^{(1)}A_2(1)\bigr)=(1,0),~
\bigl(A_3(2),\partial^{(1)}A_3(2)\bigr)=(4,2)$.
For $s=2$, Theorem~\ref{duality-thm} gives
\[
w_{1,j}(g)=
\partial^{(1)}
\left(\frac{g(x)}{A_j(x)}\right)\bigg|_{x=\alpha_j},
\qquad
w_{2,j}(g)=
\frac{g(\alpha_j)}{A_j(\alpha_j)} .
\]
Equivalently, since $A_j(\alpha_j)^2=1$ in $\mathbb F_5$, we have
\[
w_{1,j}(g)
=
\partial^{(1)}g(\alpha_j)A_j(\alpha_j)
-
g(\alpha_j)\partial^{(1)}A_j(\alpha_j),
\qquad
w_{2,j}(g)=g(\alpha_j)A_j(\alpha_j)^{-1}.
\]
Applying these formulas to the basis $\{1,x\}$ of $\mathcal P_{<2}$, we obtain
\[
W(1)=
\begin{bmatrix}
2&0&3\\
4&1&4
\end{bmatrix},
\qquad
W(x)=
\begin{bmatrix}
4&1&0\\
0&1&3
\end{bmatrix}.
\]
Thus $C^\perp=\operatorname{span}_{\mathbb F_5}\{W(1),W(x)\}$. 
One checks directly that each of $W(1)$ and $W(x)$ is orthogonal to
$M(1),M(x),M(x^2),M(x^3)$ under the Euclidean inner product. This gives a
concrete verification of the dual formula in Theorem~\ref{duality-thm}.

We finally show that this dual code is not an HRS code. Suppose otherwise that $C^\perp=\operatorname{HRS}_2(\mathcal B,U,2,3)$
for some evaluation set $\mathcal B=\{\beta_1,\beta_2,\beta_3\}$ and some
multiplier matrix $U=(u_{i,j})\in(\mathbb F_5^\ast)^{2\times 3} $.
Then the constant polynomial $1$ would give a nonzero codeword of
$C^\perp$ whose second row is identically zero, namely
\[
\begin{bmatrix}
u_{1,1}&u_{1,2}&u_{1,3}\\
0&0&0
\end{bmatrix}.
\]
However, every codeword of the present dual code has the form
$aW(1)+bW(x)$, where $a,b\in\mathbb F_5$. Its second row is
\[
a(4,1,4)+b(0,1,3)=(4a,a+b,4a+3b).
\]
If this row is zero, then $4a=0$, so $a=0$. Then $a+b=0$ gives
$b=0$. Hence the only codeword of $C^\perp$ whose second row is zero is
the zero codeword. This contradicts the existence of the above nonzero
constant-polynomial codeword. Therefore $C^\perp$ cannot be represented as any HRS code with multiplicity $2$ and three evaluation points over $\mathbb F_5$.
\end{example}

Theorem \ref{duality-thm} gives the dual code in terms of the local expansion
of $1/A_j(x)$. We now consider an important special case where the
evaluation set is the whole finite field. In this case, the local
expansion of $1/A_j(x)$ can be computed explicitly. Moreover, when
$s<q$, all nonconstant terms in this expansion are irrelevant for the
first $s$ Hasse derivatives, and the dual code takes a particularly
simple reverse-order form. We first introduce the row-reversal map
\[
\operatorname{Rev}:
\operatorname{Mat}_{s\times r}(\mathbb{F}_q)
\longrightarrow
\operatorname{Mat}_{s\times r}(\mathbb{F}_q)
\]
defined by
\[
\operatorname{Rev}\bigl((a_{i,j})\bigr)
=
(a_{s-i+1,j}).
\]

\begin{corollary}\label{Full-Evaluation-Points}
Assume that the evaluation set is the whole field $\mathcal{A}=\{\alpha_1,\ldots,\alpha_q\}=\mathbb{F}_q$, and let $C=\operatorname{HRS}_k(\mathcal{A},V,s,q).$
Then the dual code is given by $C^\perp=\{W(g)=(w_{i,j}(g)): g\in\mathcal{P}_{<sq-k}\}$, where
\[
w_{i,j}(g)
=
\frac{(-1)^s}{v_{i,j}}
\sum_{\tau=0}^{\left\lfloor (s-i)/(q-1)\right\rfloor}
\binom{s+\tau-1}{\tau}
\partial^{(s-i-\tau(q-1))}g(\alpha_j).
\]
In particular, if $s<q$, define
$V^\perp=(v^\perp_{i,j})_{1\le i\le s,\ 1\le j\le q}
$
by
$v^\perp_{i,j}=\frac{(-1)^s}{v_{s-i+1,j}}.
$
Then
\[
C^\perp
=
\operatorname{Rev}
\left(
\operatorname{HRS}_{sq-k}(\mathcal{A},V^\perp,s,q)
\right).
\]
\end{corollary}

\begin{proof}
By Theorem~\ref{duality-thm}, every codeword of $C^\perp$ is of the form
$W(g)=(w_{i,j}(g))$, where $g\in\mathcal{P}_{<sq-k}$ and $w_{i,j}(g)
=
\frac{1}{v_{i,j}}
\partial^{(s-i)}
\left(
\frac{g(x)}{A_j(x)}
\right)\bigg|_{x=\alpha_j}$.
We compute the local expansion of $1/A_j(x)$. Fix $j$, and set
$y=x-\alpha_j.$ Since the evaluation set is the whole field,
$\{\alpha_\ell-\alpha_j:\ell\ne j\}=\mathbb{F}_q^\ast.$
Hence
\[
A_j(x)
=
\prod_{\ell\ne j}(x-\alpha_\ell)^s
=
\prod_{a\in\mathbb{F}_q^\ast}(y-a)^s
=
(y^{q-1}-1)^s.
\]
Therefore
\[
\frac{1}{A_j(x)}
=
(y^{q-1}-1)^{-s}
=
(-1)^s(1-y^{q-1})^{-s}.
\]
Using the binomial expansion
$(1-T)^{-s}
=
\sum_{\tau\ge 0}\binom{s+\tau-1}{\tau}T^\tau,$
we obtain
\[
\frac{1}{A_j(x)}
=
(-1)^s
\sum_{\tau\ge 0}
\binom{s+\tau-1}{\tau}
y^{\tau(q-1)}.
\]
Substituting this expansion into Theorem~\ref{duality-thm} gives
\[
w_{i,j}(g)
=
\frac{(-1)^s}{v_{i,j}}
\sum_{\tau=0}^{\left\lfloor (s-i)/(q-1)\right\rfloor}
\binom{s+\tau-1}{\tau}
\partial^{(s-i-\tau(q-1))}g(\alpha_j),
\]
which proves the first assertion.

Now assume $s<q$. Then,
$0\le s-i\le s-1<q-1.$
Thus, only the term $\tau=0$ appears in the above summation, and hence
\[w_{i,j}(g)=\frac{(-1)^s}{v_{i,j}}\partial^{(s-i)}g(\alpha_j).\]
Applying the row-reversal map, we get
\[
\operatorname{Rev}(W(g))_{i,j}
=
W(g)_{s-i+1,j}
=
\frac{(-1)^s}{v_{s-i+1,j}}
\partial^{(i-1)}g(\alpha_j).
\]
By the definition of $V^\perp$, this becomes
\[
\operatorname{Rev}(W(g))_{i,j}
=
v^\perp_{i,j}\partial^{(i-1)}g(\alpha_j).
\]
As $g$ ranges over $\mathcal{P}_{<sq-k}$, the row-reversed dual code is
$\operatorname{HRS}_{sq-k}(\mathcal{A},V^\perp,s,q),$ which proves
\[
\operatorname{Rev}(C^\perp)
=
\operatorname{HRS}_{sq-k}(\mathcal{A},V^\perp,s,q).
\]
Since $\operatorname{Rev}$ is an involution, the equivalent expression
for $C^\perp$ follows.
\end{proof}

\begin{example}
Consider an HRS code over $\mathbb{F}_3$ with multiplicity $s = 2$, dimension $k = 3$, and $r = 3$ evaluation points $\mathcal{A} = \{0, 1, 2\}=\mathbb{F}_3$. Let $V$ be the all-one multiplier matrix. The length of the code is $N = rs = 6$.
By definition, $C = \operatorname{HRS}_3(\mathcal{A}, V, 2, 3)$ is generated by 
$$M(1) = \begin{bmatrix} 1 & 1 & 1 \\ 0 & 0 & 0 \end{bmatrix}, \quad M(x) = \begin{bmatrix} 0 & 1 & 2 \\ 1 & 1 & 1 \end{bmatrix}, \quad M(x^2) = \begin{bmatrix} 0 & 1 & 1 \\ 0 & 2 & 1 \end{bmatrix}.$$
By Corollary~\ref{Full-Evaluation-Points}, the dual code $C^\perp$ can be generated as follows. Applying  reverse-order evaluation to the dual basis $\{1, x, x^2\}$ over $\mathbb{F}_3$ yields the corresponding dual codewords: $$W(1) = \begin{bmatrix} 0 & 0 & 0 \\ 1 & 1 & 1 \end{bmatrix}, \quad W(x) = \begin{bmatrix} 1 & 1 & 1 \\ 0 & 1 & 2 \end{bmatrix}, \quad W(x^2) = \begin{bmatrix} 0 & 2 & 1 \\ 0 & 1 & 1 \end{bmatrix}.$$
Consequently, for any primal codeword $\mathbf{c}_f \in C$ and any dual codeword $\mathbf{w}_g \in C^\perp$, their standard dot product over $\mathbb{F}_3$ identically vanishes. Comparing these generators, we see that $W$ matches $M$ exactly after reversing the order of the rows in each evaluation block. Therefore, in what follows, we consider the class of reverse self-dual HRS codes.
\end{example}

\section{Reverse Self-Dual HRS Codes}\label{secIV}

The dual description in Section~\ref{secIII} shows that the Hasse derivative
orders are naturally reversed in the dual structure of HRS codes. This
motivates the following notion of self-duality, where the inner product
is taken after reversing the order of the rows in each evaluation block.

For $A,B\in \operatorname{Mat}_{s\times r}(\mathbb{F}_q),$
define the reverse bilinear form by
\[
[A,B]_{\rm rev}
:=
\langle A,\operatorname{Rev}(B)\rangle
=
\sum_{j=1}^{r}\sum_{i=1}^{s}a_{i,j}b_{s-i+1,j},
\]
where the row-reversal map $\operatorname{Rev}(\cdot)$ was defined in Section~\ref{secIII}. For a linear code $C\subseteq \operatorname{Mat}_{s\times r}(\mathbb{F}_q),$ its reverse dual is defined as
\[
C^{\perp_{\rm rev}}
:=
\left\{
B\in \operatorname{Mat}_{s\times r}(\mathbb{F}_q):
[A,B]_{\rm rev}=0\ \text{for all } A\in C
\right\}.
\]
Equivalently, $C^{\perp_{\rm rev}}
=
\operatorname{Rev}(C^\perp)$.
We say that $C$ is reverse self-orthogonal if
$C\subseteq C^{\perp_{\rm rev}},$
and reverse self-dual if
$C=C^{\perp_{\rm rev}}.$ 	Equivalently, reverse self-duality means $C^\perp=\operatorname{Rev}(C)$. In particular, a reverse self-dual code must satisfy
	\[
	\dim C=\dim C^{\perp_{\rm rev}}=sr-\dim C,
	\]
	and hence $\dim C=sr/2$.
Thus, for an HRS code $\operatorname{HRS}_k(\mathcal{A},V,s,r)$, reverse
self-duality can occur only when $sr$ is even and $k=sr/2$.

In the remainder of this section, we first study the multiplicity-two
case, where the reverse self-duality condition has a simple moment
description. We then extend the criterion and constructions to general
multiplicity.

\subsection{The Multiplicity-Two Case}
We begin with the case $s=2$. Let
$\mathcal{A}=\{\alpha_1,\ldots,\alpha_r\}\subseteq \mathbb{F}_q$ be a set of pairwise distinct evaluation points. For $j=1,2,\cdots, r$, let $v_{1,j}$ and $v_{2,j}$ be nonzero elements of $\mathbb{F}_q$, and set the multiplier matrix $V=\begin{pmatrix} v_{1,j} \\ v_{2,j}\end{pmatrix}_{1 \leq j \leq r}\in (\mathbb{F}^*_q)^{2\times r}$. In this subsection, we consider the HRS code
\[
C=\operatorname{HRS}_r(\mathcal{A},V,2,r).
\]
In this case, the reverse bilinear form becomes
\[
[x,y]_{\rm rev}
=
\sum_{j=1}^{r}
\bigl(x_{1,j}y_{2,j}+x_{2,j}y_{1,j}\bigr).
\]
The following theorem reduces reverse self-duality to the vanishing of a
finite sequence of weighted power sums over the evaluation set.

\begin{theorem}\label{thm:criterion}
Set $\lambda_j:=v_{1,j}v_{2,j}$ for $1\le j\le r.$
The HRS code $C=\operatorname{HRS}_r(\mathcal{A},V,2,r)$ is reverse self-dual if and only if
\begin{equation}\label{eq:criterion1}
(m+1)\sum_{j=1}^{r}\lambda_j\alpha_j^m=0,
\qquad
0\le m\le 2r-3.
\end{equation}
\end{theorem}

\begin{proof}
 For each $0\le a\le r-1$, define
		\[
	c^{(a)}:=
	\left(\begin{array}{c}
		v_{1,j}\partial^{(0)}x^a(\alpha_j)\\
		v_{2,j}\partial^{(1)}x^a(\alpha_j)
	\end{array}\right)_{1 \leq j \leq r}.
	\]
    	Here
	\[
	\partial^{(1)}x^a(\alpha_j)=
	\begin{cases}
		a\alpha_j^{a-1}, & a\ge 1,\\
		0, & a=0.
	\end{cases}
	\]
	These vectors span $C$. Since $\dim(C)=r$, $C$ is reverse self-dual if and only if it is reverse self-orthogonal, which is also equivalent to
	\[
	[c^{(a)},c^{(b)}]_\mathrm{rev}=0,
	\qquad (0\le a,b\le r-1).
	\]
Using the product rule for Hasse derivatives, we obtain
	\begin{align*}
		[c^{(a)},c^{(b)}]_{\rm rev}
		&=\sum_{j=1}^r\lambda_j\Big(
		\partial^{(0)}(x^a)(\alpha_j)\partial^{(1)}(x^b)(\alpha_j)
		+\partial^{(1)}(x^a)(\alpha_j)\partial^{(0)}(x^b)(\alpha_j)\Big)\\
		&=\sum_{j=1}^r\lambda_j\partial^{(1)}(x^{a+b})(\alpha_j).
		\end{align*}
	If $a+b=0$, this expression is zero. If $a+b\ge1$, then
		\[
		[c^{(a)},c^{(b)}]_{\rm rev}
		=(a+b)\sum_{j=1}^{r}\lambda_j\alpha_j^{a+b-1}.
		\]
		Putting $m=a+b-1$ gives the identities in \eqref{eq:criterion1}, where $0\le m\le 2r-3$. Conversely, for every $0\le m\le 2r-3$, choose $0\le a,b\le r-1$ such that $a+b=m+1$. Hence the moment identities in \eqref{eq:criterion1} are necessary and sufficient for reverse self-orthogonality, and therefore for reverse self-duality.
\end{proof}

The following lemma is useful in our constructions.

\begin{lemma}\label{lem:moment}
	Let
	$	H(x)=\prod_{i=1}^{r}(x-\alpha_i)= x^r+h_{r-1}x^{r-1}+\cdots+h_1x+h_0.$
For any $m\ge 0$, define weighted moments
	\[
	S_m:=\sum_{j=1}^r\frac{\alpha_j^m}{H'(\alpha_j)}.
	\]
	Then
    \begin{align}\label{eq:lowmom}
    S_m = \begin{cases} 0, & 0 \le m \le r-2, \\ 1, & m = r-1. \end{cases},
    \end{align}
	and for every $m\ge r$,
	\begin{equation}\label{eq:recur-general}
		S_m+h_{r-1}S_{m-1}+h_{r-2}S_{m-2}+\cdots+h_1S_{m-r+1}+h_0S_{m-r}=0.
	\end{equation}
\end{lemma}

\begin{proof}
	By the Lagrange interpolation formula, for any $0 \leq m \leq r-1$,
    \[x^m=\sum_{j=1}^r\frac{\prod_{i \neq j}(x-\alpha_i)}{H'(\alpha_j)}\alpha_j^m.\]
    The identities in \eqref{eq:lowmom}  follow by comparing the coefficients of the term $x^{r-1}$ in both sides of the above formula.
	
	We now derive the recurrence \eqref{eq:recur-general} in detail. Since each $\alpha_j$ is a root of $H$, we have
	\[
	\alpha_j^r+h_{r-1}\alpha_j^{r-1}+\cdots+h_1\alpha_j+h_0=0
	\qquad (1\le j\le r).
	\]
	Then, for any $m\ge r$,
	\[
	\frac{\alpha_j^m}{H'(\alpha_j)}
	+h_{r-1}\frac{\alpha_j^{m-1}}{H'(\alpha_j)}
	+\cdots
	+h_1\frac{\alpha_j^{m-r+1}}{H'(\alpha_j)}
	+h_0\frac{\alpha_j^{m-r}}{H'(\alpha_j)}
	=0.
	\]
	Summing over all $j=1,\dots,r$ gives
	\[
	S_m+h_{r-1}S_{m-1}+h_{r-2}S_{m-2}+\cdots+h_1S_{m-r+1}+h_0S_{m-r}=0,
	\]
	which is precisely \eqref{eq:recur-general}.
\end{proof}
We now give several explicit families satisfying the moment conditions in
Theorem~\ref{thm:criterion}. The first is based on additive cosets.

\begin{theorem}[Additive coset construction for $s=2$]\label{thm:additive-single}
	Let $r=p^t$ for some $t\geq 1$,  and assume that $\mathbb{F}_r$ is a subfield of $ \mathbb{F}_q$. Fix $a\in\mathbb{F}_q$ and $\beta\in\mathbb{F}_q^*$.
	Choose $\lambda\in\mathbb{F}_q^*$ and arbitrary nonzero scalars $u_b\in\mathbb{F}_q^*$ for $b\in\mathbb{F}_r$.  Let $\mathcal{A}=a+\beta\mathbb{F}_r=\{\alpha_b=a+\beta b: b\in\mathbb{F}_r\}$.  Define the multiplier matrix
	\[
	V=\begin{pmatrix} v_{1,b} \\ v_{2,b}\end{pmatrix}_{b \in \mathbb{F}_r}=\begin{pmatrix} u_b \\ \lambda u_b^{-1}\end{pmatrix}_{b \in \mathbb{F}_r} \in (\mathbb{F}^*_q)^{2\times r}.
	\]
	Then the HRS code $C=\operatorname{HRS}_r(\mathcal{A},V, 2,r)$ is reverse self-dual.
\end{theorem}

\begin{proof}
	The vanishing polynomial of $\mathcal{A}$ is
	\[
	H(x):=\prod_{b \in \mathbb{F}_r}(x-\alpha_b)= (x-a)^r-\beta^{r-1}(x-a)=x^r-\beta^{r-1}x+\big(\beta^{r-1}a-a^r\big).
	\]
  Thus $H'(x)=-\beta^{r-1}$.
Then the weighted moments
	\[
	S_m=\sum_{j=1}^r\frac{\alpha_j^m}{H'(\alpha_j)}=-\beta^{1-r}\sum_{b\in\mathbb{F}_r}\alpha_b^m.
	\]
	By Lemma \ref{lem:moment}, $S_m=0$ for $0\le m\le r-2$
and
	\begin{equation}\label{eq:add-single-recur}
		S_m
		=
		\beta^{r-1}S_{m-r+1}
		-
		\big(\beta^{r-1}a-a^r\big)S_{m-r}.\qquad (m\ge r)
	\end{equation}
	When $r\le m\le 2r-3$, then $0\le m-r\le r-3$ and $1\le m-r+1\le r-2.$
	Therefore, the right-hand side of \eqref{eq:add-single-recur} is zero, and so
	\[
	S_m=0\qquad (r\le m\le 2r-3).
	\]
	Note that when $m=r-1$, then $m+1=r$ vanishes in characteristic $p$.
	Since $v_{1,b}v_{2,b}=\lambda$ is a nonzero constant, the criterion \eqref{eq:criterion1} in Theorem~\ref{thm:criterion} follows, which completes the proof.
\end{proof}

\begin{example}
Consider an HRS code over $\mathbb{F}_{16}$ with multiplicity $s = 2$, dimension $k = 4$, and $r = 4$. Let $\mathbb F_4=\{0,1,\omega,\omega^2\}$ be the subfield of $\mathbb F_{16}$, where $\omega^2+\omega+1=0$. We define the evaluation set $\mathcal{A} = \{ \alpha_b := 1 + b \mid b \in \mathbb{F}_4 \}=\{1,0,1+\omega, 1+\omega^2\}$. 
Choose $\lambda=1$ and set
$(u_0,u_1,u_\omega,u_{\omega^2})=(1,1,\omega,\omega^2)$. By Theorem~\ref{thm:additive-single}, the multiplier matrix is
$$V = \begin{bmatrix} 1 & 1 & \omega & \omega+1 \\ 1 & 1 & \omega+1 & \omega \end{bmatrix}.$$
Then the HRS code $C = \operatorname{HRS}_4(\mathcal{A}, V, 2, 4)$ is generated by 
$$M(1) = \begin{bmatrix} 1 & 1 & \omega & \omega+1 \\ 0 & 0 & 0 & 0 \end{bmatrix}, \quad M(x) = \begin{bmatrix} 1 & 0 & 1 & 1 \\ 1 & 1 & \omega+1 &\omega \end{bmatrix},$$
$$M(x^2) = \begin{bmatrix} 1 & 0 & \omega+1 & \omega \\ 0 & 0 & 0 & 0 \end{bmatrix}, \quad M(x^3) = \begin{bmatrix} 1 & 0 & \omega & \omega+1 \\ 1 & 0 & 1 & 1 \end{bmatrix}.$$
Let $G_a=M(x^a)$ for $0\le a\le 3$. Then
$C=\operatorname{span}_{\mathbb F_{16}}\{G_0,G_1,G_2,G_3\}$, and
	\[
	\operatorname{Rev}(C)
	=
	\operatorname{span}_{\mathbb F_{16}}
	\{\operatorname{Rev}(G_0),\operatorname{Rev}(G_1),
	\operatorname{Rev}(G_2),\operatorname{Rev}(G_3)\}.
	\]
A direct computation gives
	\[
	\left\langle G_a,\operatorname{Rev}(G_b)\right\rangle=0,
	\qquad 0\le a,b\le 3.
	\]
Thus $\operatorname{Rev}(C)\subseteq C^\perp$. Then $C^\perp=\operatorname{Rev}(C)$ since $\dim(\operatorname{Rev}(C))=\dim C=\dim C^\perp=4$. Equivalently 
$C=C^{\perp_{\rm rev}}$. This verifies the reverse self-duality asserted by Theorem~\ref{thm:additive-single}.
\end{example}

We next turn to multiplicative cosets.

\begin{theorem}[Multiplicative subgroup coset construction for $s=2$]\label{thm:mult-single}
	Let $G\subseteq\mathbb{F}_q^*$ be a multiplicative subgroup of order $r$, and let $\mathcal{A}=\beta G$ for some $\beta\in\mathbb{F}_q^*$. Assume that
	\[
	p:=\operatorname{char}(\mathbb{F}_q)\mid (r-1).
	\]
	Choose $\mu\in\mathbb{F}_q^*$ and arbitrary nonzero scalars $u_\alpha\in\mathbb{F}_q^*$ for $\alpha\in \mathcal{A}$. Define the multiplier matrix
	\[V=\begin{pmatrix} v_{1,\alpha} \\ v_{2,\alpha}\end{pmatrix}_{\alpha \in \mathcal{A}}=\begin{pmatrix} u_\alpha \\ \mu\alpha^2u_\alpha^{-1}\end{pmatrix}_{\alpha \in \mathcal{A}}\in (\mathbb{F}^*_q)^{2\times r}.
	\]
	Then the HRS code $C=\operatorname{HRS}_r(\mathcal{A},V, 2,r)$ is reverse self-dual.
\end{theorem}
\begin{proof}
The vanishing polynomial of $\mathcal{A}$ is $H(x)=\prod_{\alpha\in \mathcal{A}}(x-\alpha)=x^r-\beta^r$.
Hence, $H'(x)=r x^{r-1}$.
Since $r\mid(q-1)$, we have $p\nmid r$, and therefore $r\ne 0$ in
$\mathbb F_q$. For every $\alpha\in \mathcal{A}$, one has
$\alpha^r=\beta^r$, and thus $H'(\alpha)=r\alpha^{r-1}
        =\frac{r\beta^r}{\alpha}$.
Consequently,
 \[
        \sum_{\alpha\in \mathcal{A}}\alpha^{m+2}
        =
        r\beta^r\sum_{\alpha\in \mathcal{A}}\frac{\alpha^{m+1}}{H'(\alpha)}=r\beta^rS_{m+1}.
\]
By Lemma~\ref{lem:moment}, applied to
$H(x)=x^r-\beta^r$, we have
$$S_m = \begin{cases} 0, & 0 \le m \le r-2, \\ 1, & m = r-1. \end{cases}$$
Moreover, for every $m\ge r$,
\[
        S_m-\beta^r S_{m-r}=0.
\]
It follows that
\[
        S_m=0,\qquad m\in \{0,1,\ldots,r-2\}\cup\{r,r+1,\ldots,2r-2\},
\]
and the only possibly nonzero value among $0\le m\le 2r-2$ is
$   S_{r-1}=1.$
Then
\[
       (m+1)\sum_{\alpha\in \mathcal{A}}\lambda_{\alpha}\alpha^m=(m+1)\sum_{\alpha\in \mathcal{A}}v_{1,\alpha}v_{2,\alpha}\alpha^m =(m+1)\mu\sum_{\alpha\in \mathcal{A}}\alpha^{m+2}=(m+1)\mu r\beta^r S_{m+1}.
\]
The preceding identities show that this expression is zero for $0 \leq m \leq 2r-3$ and $m \neq r-2$. When $m= r-2$, then $m+1=r-1=0$
in $\mathbb F_q$. Hence, the desired identity
also holds in this exceptional case.

Therefore, the code is reverse
self-dual by the moment criterion in
Theorem~\ref{thm:criterion}.
\end{proof}

We now present a construction of a reverse self-dual HRS code by applying the multiplicative subgroup coset construction described in Theorem~\ref{thm:mult-single}.
\begin{example}
Let $\mathbb{F}_4 = \{0, 1, \omega, \omega^2\}$, where $\omega^2 + \omega + 1 = 0$. Let $G = \mathbb{F}_4^* = \{1, \omega, \omega^2\}$ be the cyclic multiplicative subgroup of order $r = 3$. Then, it is obvious that $p \mid (r - 1)$. Take $\beta = 1 \in \mathbb{F}_4^*$, then
$\mathcal{A} = \beta G = \{\alpha_1 = 1, \; \alpha_2 = \omega, \; \alpha_3 = \omega^2\}$.
Choose $\mu = 1 \in \mathbb{F}_4^*$ and assign the local nonzero scalars
$u_\alpha\in\mathbb F_4^*$ for $\alpha\in\mathcal A$ by
$(u_1,u_\omega,u_{\omega^2})=(1,1,\omega).$ Then the multiplier matrix is
$$V = \begin{pmatrix} v_{1, \alpha} \\ v_{2, \alpha} \end{pmatrix}_{\alpha \in \mathcal{A}} = \begin{bmatrix} 1 & 1 & \omega \\ 1 & \omega^2 & 1 \end{bmatrix}.$$
The total code length is $N = rs = 3 \times 2 = 6$. Then the HRS code $C = \operatorname{HRS}_3(\mathcal{A}, V, 2, 3)$ is  spanned by 
$$M(1) = \begin{bmatrix} 1 & 1 & \omega \\ 0 & 0 & 0 \end{bmatrix}, \quad M(x) = \begin{bmatrix} 1 & \omega & 1 \\ 1 & \omega^2 & 1 \end{bmatrix},\quad M(x^2) = \begin{bmatrix} 1 & \omega^2 & \omega^2 \\ 0 & 0 & 0 \end{bmatrix}.$$
Next, we verify the conclusion of Theorem~\ref{thm:mult-single} directly. Let $G_a=M(x^a)$ for $0\le a\le 2$. Then $C=\operatorname{span}_{\mathbb F_{4}}\{G_0,G_1,G_2\}$, and its row-reversal space is
\[\operatorname{Rev}(C) = \operatorname{span}_{\mathbb F_{4}}\{\operatorname{Rev}(G_0),\operatorname{Rev}(G_1),\operatorname{Rev}(G_2)\}.\] 
A direct computation of the Euclidean inner product gives 
\[
\langle G_a, \operatorname{Rev}(G_b) \rangle = 0 \quad \text{for all } 0 \le a, b \le 2.
\]
Thus, $\operatorname{Rev}(C)\subseteq C^\perp$. Since $\dim \operatorname{Rev}(C)=\dim C=3$ and $\dim C^\perp = 6-3=3$, we must have $C^\perp=\operatorname{Rev}(C)$, or equivalently $C=C^{\perp_{\rm rev}}$. This confirms that $C$ is reverse self-dual as asserted.
\end{example}

\begin{theorem}[Multiplicative subgroup with zero for $s=2$]\label{thm:mult-s=2}
Let $G\le \mathbb F_q^*$ be a multiplicative subgroup of order $r-1$,
and let $\beta\in\mathbb F_q^*$. Let
$   \mathcal{A}=\{0\}\cup \beta G,$
so that $|\mathcal{A}|=r$. Assume that
\[
        p=\operatorname{char}(\mathbb F_q)\mid r .
\]
Choose $\lambda\in\mathbb F_q^*$ and arbitrary nonzero scalars
$u_\alpha\in\mathbb F_q^*$, $\alpha\in \mathcal{A}$. Define the multiplier matrix
	\[
	V=\begin{pmatrix} v_{1,\alpha} \\ v_{2,\alpha}\end{pmatrix}_{\alpha\in \mathcal{A}}=\begin{pmatrix} u_{\alpha} \\ \lambda u_{\alpha}^{-1}\end{pmatrix}_{\alpha\in \mathcal{A}}\in (\mathbb{F}^*_q)^{2\times r}.
	\]
	Then the HRS code $C=\operatorname{HRS}_r(\mathcal{A},V, 2,r)$ is reverse self-dual.
\end{theorem}
\begin{proof}
Since $\mathcal{A}=\{0\}\cup \beta G$ and $|G|=r-1$, the vanishing polynomial of $\mathcal{A}$ is $H(x) = x(x^{r-1}-\beta^{r-1}) = x^{r}-\beta^{r-1}x$. Since $r=0$ in $\mathbb{F}_q$, the derivative $H'(x) = -\beta^{r-1}$ is a nonzero constant.

For $m\ge 0$, define the weighted moments
$   S_m:=\sum_{\alpha\in  \mathcal{A}}\frac{\alpha^m}{H'(\alpha)}.$
By Lemma~\ref{lem:moment}, 
we have
$$S_m = \begin{cases} 0, & 0 \le m \le r-2, \\ 1, & m = r-1. \end{cases}$$
and
\[
        S_m-\beta^{r-1} S_{m-r+1}=0,\qquad m\ge r.
\]
Thus, for $r\le m\le 2r-3,$
we have $1\le m-r+1\le r-2,$
and hence
\[
        S_m=0,\qquad r\le m\le 2r-3.
\]
Since $H'(\alpha)=-\beta^{r-1}$ for all $\alpha\in \mathcal{A}$, we have 
\[\sum_{\alpha\in  \mathcal{A}}\alpha^m
        =
        -\beta^{r-1} S_m =0\]
for all $0\le m\le 2r-3$ with $m \neq r-1$. Moreover, when $m=r-1$, then $m+1=r=0$
in $\mathbb F_q$. Hence
\[
        (m+1)\sum_{\alpha\in  \mathcal{A}}\lambda\alpha^m=0,
        \qquad 0\le m\le 2r-3.
\]
Since $v_{1,\alpha}v_{2,\alpha}=\lambda$ for all $\alpha\in  \mathcal{A}$,
Theorem~\ref{thm:criterion} implies that $C$ is reverse self-dual.
\end{proof}

\begin{example}
Let $\mathbb{F}_4 = \{0, 1, \omega, \omega^2\}$, where $\omega^2 + \omega + 1 = 0$. Let $G = \{1\} \subseteq \mathbb{F}_4^*$ be the trivial multiplicative subgroup of order $|G| = r - 1 = 1$. This choice satisfies the divisibility condition $p \mid r$.
Choose $\beta=1\in\mathbb F_4^*$. Then $\mathcal A=\{0\}\cup \beta G=\{0,1\}$. Let $\lambda = \omega \in \mathbb{F}_4^*$ and  $(u_0, u_1) = (1, \omega^2)$. Then the multiplier matrix  $V $ is
$$V =\begin{pmatrix} u_\alpha \\ \lambda u_\alpha^{-1} \end{pmatrix}_{\alpha \in \mathcal{A}} = \begin{bmatrix} 1 & \omega^2 \\ \omega & \omega^2 \end{bmatrix}.$$
Setting the dimension of the code to $k = 2$, the total length of the code is $N = rs = 2 \times 2 = 4$. Then the HRS code $C = \operatorname{HRS}_2(\mathcal{A}, V, 2, 2)$ is spanned by 
$$M(1) = \begin{bmatrix} 1 &  \omega^2 \\ 0 & 0  \end{bmatrix}, \quad M(x) = \begin{bmatrix}0 & \omega^2  \\ \omega & \omega^2  \end{bmatrix}.$$
Denote $G_a = M(x^a)$ for $0 \le a \le 1$, so that $C = \operatorname{span}_{\mathbb{F}_4}\{G_0, G_1\}$. Then
\[
\operatorname{Rev}(C) = \operatorname{span}_{\mathbb{F}_4}\{\operatorname{Rev}(G_0), \operatorname{Rev}(G_1)\}.
\]
A direct computation shows that $\langle G_a, \operatorname{Rev}(G_b) \rangle = 0$ for all $0 \le a, b \le 1$. Thus $\operatorname{Rev}(C) \subseteq C^\perp$. Since $\dim \operatorname{Rev}(C) = 2 = \dim C^\perp$, we obtain $C^\perp = \operatorname{Rev}(C)$,
which confirms that $C = C^{\perp_{\rm rev}}$ in agreement with Theorem~\ref{thm:mult-s=2}.
\end{example}

\subsection{The General Multiplicity Case}
We now extend the preceding discussion to arbitrary multiplicity $s$.
The main difference is that the reverse bilinear form involves all pairs
of Hasse derivative orders whose sum is $s-1$. Under a natural symmetry
condition on the multipliers, the criterion again reduces to weighted
power-sum identities.

We first introduce the well-known Lucas' Theorem. Let $p$ be a prime. For nonnegative integers $n$ and $k$, write their base-$p$ expansions as
\[
 n=n_0+n_1p+\cdots+n_Lp^L, \qquad
 k=k_0+k_1p+\cdots+k_Lp^L,
\]
where $0\le n_i,k_i\le p-1$ for all $i$. We allow leading zeros so that both expansions have the same length.

\begin{theorem}[Lucas' Theorem \cite{Lucas78}]
With the notation above, one has
\[
 \binom{n}{k}\equiv \prod_{i=0}^L \binom{n_i}{k_i} \pmod p.
\]
Here we use the convention that $\binom{a}{b}=0$ if $b>a$.
\end{theorem}

	Before stating the general criterion, let $k=rs/2$ and define
	\[
	c^{(a)}=\operatorname{Ev}_{\mathcal A,V}(x^a),\qquad 0\le a\le k-1.
	\] 
	\begin{theorem}\label{thm:general_criterion}
		Suppose $rs$ is even, and let
		$\mathcal{A}=\{\alpha_1,\ldots,\alpha_r\}\subseteq\mathbb{F}_q$
		be a set of $r$ pairwise distinct evaluation points. Let
		$V=(v_{i,j})\in(\mathbb F_q^*)^{s\times r}$ be a multiplier matrix.
		Assume that there exist $\lambda_1,\ldots,\lambda_r\in\mathbb F_q^*$ such that
		\[
		v_{t+1,j}v_{s-t,j}=\lambda_j,
		\qquad
		0\le t\le s-1,
		\quad 1\le j\le r.
		\]
		Then the HRS code
	$	C=\operatorname{HRS}_{\frac{rs}{2}}(\mathcal{A},V,s,r)$
		is reverse self-dual if and only if
		\begin{equation}\label{eq:general_criterion}
			\binom{m+s-1}{s-1}
			\sum_{j=1}^r \lambda_j\alpha_j^m=0,
			\qquad 0\le m\le rs-s-1,
		\end{equation}
		where the binomial coefficient is viewed as an element of $\mathbb F_q$.
	\end{theorem}
	
	\begin{proof}
		Set $k=rs/2$. Since $\dim C=k$ and $\dim C^{\perp_{\rm rev}}=rs-k=k$, the code $C$ is reverse self-dual if and only if it is reverse self-orthogonal. Thus it suffices to test the reverse inner products of the monomial evaluation vectors $c^{(a)}$, $0\le a\le k-1$.
		
		For $0\le a,b\le k-1$, the symmetric multiplier condition and the product rule for Hasse derivatives give
		\begin{align*}
			[c^{(a)},c^{(b)}]_{\rm rev}
			&=\sum_{j=1}^r\sum_{i=1}^s
			v_{i,j}v_{s-i+1,j}
			\partial^{(i-1)}(x^a)(\alpha_j)
			\partial^{(s-i)}(x^b)(\alpha_j)\\
			&=\sum_{j=1}^r\lambda_j
			\sum_{t=0}^{s-1}
			\partial^{(t)}(x^a)(\alpha_j)
			\partial^{(s-1-t)}(x^b)(\alpha_j)\\
			&=\sum_{j=1}^r\lambda_j
			\partial^{(s-1)}(x^{a+b})(\alpha_j).
		\end{align*}
		This formula is defined even when some $\alpha_j=0$. If $a+b<s-1$, then $\partial^{(s-1)}(x^{a+b})=0$, so $[c^{(a)},c^{(b)}]_{\rm rev}=0$. If $a+b\ge s-1$, then
		\[
		[c^{(a)},c^{(b)}]_{\rm rev}
		=
		\binom{a+b}{s-1}
		\sum_{j=1}^r\lambda_j\alpha_j^{a+b-s+1}.
		\]
		Putting $m=a+b-s+1$ gives
		\[
		[c^{(a)},c^{(b)}]_{\rm rev}
		=
		\binom{m+s-1}{s-1}
		\sum_{j=1}^r\lambda_j\alpha_j^m,
		\]
		where $0\le m\le 2k-s-1=rs-s-1$.
		
		Therefore, the moment identities in \eqref{eq:general_criterion} imply that all reverse inner products vanish, and hence $C$ is reverse self-dual. Conversely, assume that $C$ is reverse self-dual, hence reverse self-orthogonal. For each $m$ with $0\le m\le rs-s-1$, set $n=m+s-1$. Then $s-1\le n\le 2k-2$, so there exist $a,b$ with $0\le a,b\le k-1$ and $a+b=n$; for instance, take $a=\min\{n,k-1\}$ and $b=n-a$. Applying reverse self-orthogonality to this pair gives exactly the identity in \eqref{eq:general_criterion}. This proves the equivalence.
	\end{proof}

\begin{theorem}[Additive affine construction for $s=2^u$]\label{thm:additive-p-power}
Suppose $q$ is a power of $2$. Let $U\le \mathbb{F}_q$ be an $\mathbb{F}_2$-linear subspace of dimension $d$, and let
\[
  \mathcal{A}=a+U\subseteq \mathbb{F}_q,
  \qquad r=|\mathcal{A}|=2^d.
\]
Let $s=2^u$ with $u\ge 1$. 
 Choose $\lambda\in\mathbb{F}_q^*$. Define the multiplier matrix $V$ with
 \[
  v_{i,\alpha}v_{s-i+1,\alpha}=\lambda,
  \qquad 1\le i\le s,
  \quad \alpha\in \mathcal{A}.
\]
 	Then the HRS code $C=\operatorname{HRS}_{\frac{rs}{2}}(\mathcal{A},V, s,r)$ is reverse self-dual.
\end{theorem}

\begin{proof}
By Theorem~\ref{thm:general_criterion}, it suffices to show that
\[
  \binom{m+s-1}{s-1}\sum_{\alpha\in \mathcal{A}}\alpha^m=0,
  \qquad 0\le m\le rs-s-1.
\]
If the binomial coefficient is zero in $\mathbb{F}_q$, the claim is immediate. Otherwise, since $s=2^u$, we have
\[
  s-1=2^u-1=1+2+\cdots+2^{u-1}.
\]
Lucas' theorem implies that $\binom{m+s-1}{s-1}\ne0$ only if the $u$ least significant binary digits of $m+s-1$ are all $1$, or equivalently
\[
  m\equiv 0\pmod {2^u}.
\]
Thus we only need to consider that $m=\ell s$ for some integer $\ell\ge0$. For an affine $\mathbb{F}_2$-subspace $\mathcal{A}$ of size $r$, one has
\[
  \sum_{\alpha\in \mathcal{A}}\alpha^j=0,
  \qquad 0\le j\le r-2.
\]
Indeed, if $H(x)=\prod_{\alpha\in \mathcal{A}}(x-\alpha)$, then $H'(x)$ is a nonzero constant, and the claim follows by Lemma \ref{lem:moment}. Since $\ell=\frac{m}{s}\le \frac{rs-s-1}{s}=r-1-\frac{1}{s}$, 
we have $\ell\le r-2$. Therefore, using the Frobenius endomorphism,
\[
  \sum_{\alpha\in \mathcal{A}}\alpha^m
  =\sum_{\alpha\in \mathcal{A}}\alpha^{\ell s}
  =\sum_{\alpha\in \mathcal{A}}(\alpha^\ell)^{2^u}
  =\left(\sum_{\alpha\in \mathcal{A}}\alpha^\ell\right)^{2^u}
  =0.
\]
The criterion in Theorem~\ref{thm:general_criterion} now gives $C=C^{\perp_{\rm rev}}$.
\end{proof}

\begin{example}
Consider an HRS code over $\mathbb{F}_8$ with multiplicity $s = 4$, dimension $k = 4$ and $r = 2$ evaluation points. Let $U = \mathbb{F}_2= \{0, 1\} \leq \mathbb{F}_8$. We define the evaluation set as the identity affine shift $\mathcal{A} = 0 + U = \{ \alpha_b := b \mid b \in U \}=\{0, 1\}$. 
Choose $\lambda=1$ and take the multiplier matrix $V$ to be the all-one matrix. 
Then the HRS code $C = \operatorname{HRS}_4(\mathcal{A}, V, 4, 2)$ is generated by
$$M(1) = \begin{bmatrix} 1 & 1 \\ 0 & 0 \\ 0 & 0 \\ 0 & 0 \end{bmatrix}, \quad M(x) = \begin{bmatrix} 0 & 1 \\ 1 & 1 \\ 0 & 0 \\ 0 & 0 \end{bmatrix},\quad 
M(x^2) = \begin{bmatrix} 0 & 1 \\ 0 & 0 \\ 1 & 1 \\ 0 & 0 \end{bmatrix}, \quad M(x^3) = \begin{bmatrix} 0 & 1 \\ 0 & 1 \\ 0 & 1 \\ 1 & 1 \end{bmatrix}.$$
We verify the conclusion directly. Denote $G_a = M(x^a)$ for $0 \le a \le 3$, then 
\[
\operatorname{Rev}(C) = \operatorname{span}_{\mathbb{F}_8}\{\operatorname{Rev}(G_0), \operatorname{Rev}(G_1), \operatorname{Rev}(G_2), \operatorname{Rev}(G_3)\}.
\]
A direct computation shows that  $\langle G_a, \operatorname{Rev}(G_b) \rangle = 0 \quad \text{for all } 0 \le a, b \le 3$.
Thus, $\operatorname{Rev}(C) \subseteq C^\perp$. Then $C^\perp = \operatorname{Rev}(C)$ since $\dim \operatorname{Rev}(C) = \dim C = \dim C^\perp =4$.
Consequently, $C = C^{\perp_{\rm rev}}$, which verifies the reverse self-duality asserted in Theorem~\ref{thm:additive-p-power}.
\end{example}

The following lemma provides the key ingredient for extending
Theorem~\ref{thm:additive-p-power} to broader choices of $s$.

\begin{lemma}[Affine $\mathbb{F}_p$-subspace moment identity]\label{lem:affine-moment}
Let $U\le \mathbb{F}_q$ be an $\mathbb{F}_p$-linear subspace, and let $  \mathcal{A}=a+U\subseteq\mathbb{F}_q$ be an affine translation of size $r$. If $s=p^u t$ with $u\ge0$ and $1\le t\le p-1$, then for every integer $m$ with $0\le m\le rs-s-1$,
\begin{equation}\label{eq:affine-moment}
  \binom{m+s-1}{s-1}\sum_{\alpha\in \mathcal{A}}\alpha^m=0.
\end{equation}
In particular, this holds for all $2\le s\le p$.
\end{lemma}

\begin{proof}
Let $  H(x)=\prod_{\alpha\in \mathcal{A}}(x-\alpha)$
be the vanishing polynomial of $\mathcal{A}$. Let $L_U(y)=\prod_{u\in U}(y-u)$, then $H(x)=L_U(x-a)$. Since $U$ is an $\mathbb{F}_p$-linear subspace,  $L_U$ is a $p$-linearized polynomial. Hence $H'(x)=\eta\in\mathbb{F}_q^*$
is a nonzero constant, and for every $\alpha\in \mathcal{A}$ the local expansion has the form
\[
  H(\alpha+z)=L_U(z)=\eta z-R(z),
  \qquad R(z)\in z^p\mathbb{F}_q[z].
\]
Set
\[
  T(z)=\frac{R(z)}{\eta z}\in z^{p-1}\mathbb{F}_q[z].
\]
Then
\[
  H(\alpha+z)^s=\eta^s z^s(1-T(z))^s.
\]
Note that $
  (1-T)^{-s}=(1-T)^{-p^u t}=(1-T^{p^u})^{-t}=(\sum_{\ell \geq 0}T^{\ell p^u})^t$.
As $T^{p^u}\in z^{p^u(p-1)}\mathbb{F}_q[z]$ and $t\le p-1$, we get
$p^u(p-1)\ge p^u t=s$,
and therefore
\[
  (1-T)^{-s} \in 1+z^s\mathbb{F}_q[[z]],
\]
i.e., $(1-T(z))^{-s}=1+z^sf(z)$ for some $f(z) \in \mathbb{F}_q[[z]]$.
For every polynomial $h\in\mathbb{F}_q[x]$ with $\deg h\le rs-2$, we use
the Hasse--Taylor expansion
\[
  h(\alpha+z)=
  \sum_{j\ge0}\partial^{(j)}h(\alpha)z^j.
\]
Combining this with the expansion of $H(\alpha+z)^{-s}$ gives
\[
  \frac{h(\alpha+z)}{H(\alpha+z)^s}
  =
  \eta^{-s}z^{-s}
  \left(\sum_{j\ge0}\partial^{(j)}h(\alpha)z^j\right)
  (1+z^sf(z))=\eta^{-s}\left(\sum_{j\ge0}\partial^{(j)}h(\alpha)z^j\right)
\big(z^{-s}+f(z)\big).
\]
Therefore,
\[
  \operatorname{Res}_{x=\alpha}\frac{h(x)}{H(x)^s}
  =\eta^{-s}\partial^{(s-1)}h(\alpha).
\]
The finite poles of $h(x)/H(x)^s$ are precisely the points of $\mathcal A$.
Since $\deg H=r$, the denominator $H(x)^s$ has degree $rs$. The assumption
$\deg h\le rs-2$ implies that the residue at infinity is zero. By the residue
theorem on $\mathbb P^1$,
\[
  0=
  \sum_{\alpha\in\mathcal A}
  \operatorname{Res}_{x=\alpha}
  \frac{h(x)}{H(x)^s}
  =
  \eta^{-s}\sum_{\alpha\in\mathcal A}\partial^{(s-1)}h(\alpha).
\]
Since $\eta\ne0$, we obtain $\sum_{\alpha\in\mathcal A}\partial^{(s-1)}h(\alpha)=0$
for every $h\in\mathbb F_q[x]$ with $\deg h\le rs-2$.

Finally, take  $h(x)=x^{m+s-1}$.
The range $0\le m\le rs-s-1$ gives $\deg h=m+s-1\le rs-2,$ so the preceding identity applies. Since $\partial^{(s-1)}x^{m+s-1}
  =
  \binom{m+s-1}{s-1}x^m$,
we get
\[
  \binom{m+s-1}{s-1}
  \sum_{\alpha\in\mathcal A}\alpha^m=0.
\]
This proves \eqref{eq:affine-moment}.
\end{proof}

\begin{theorem}[Affine additive coset construction for $s=p^ut$]\label{thm:affine-additive-construction}
Let $U\le\mathbb{F}_q$ be an $\mathbb{F}_p$-linear subspace and let $\mathcal{A}=a+U\subseteq\mathbb{F}_q$ be an affine translation of size $r$. 
Assume that $s=p^u t$ with $u\ge0, 1\le t\le p-1,$
and that $rs$ is even. Choose $\lambda\in\mathbb{F}_q^*$. If $s$ is odd, assume in addition that $\lambda$ is a square in $\mathbb{F}_q$. Define the multiplier matrix $V$ with
 \[
  v_{i,\alpha}v_{s-i+1,\alpha}=\lambda,
  \qquad 1\le i\le s,
  \quad \alpha\in \mathcal{A}.
\]
Then the HRS code $C=\operatorname{HRS}_{\frac{rs}{2}}(\mathcal{A},V, s,r)$ is reverse self-dual.
\end{theorem}

\begin{proof}
 By construction,
$  v_{i,\alpha}v_{s-i+1,\alpha}=\lambda$
for every $1\le i\le s$ and every $\alpha\in \mathcal{A}$. Thus the symmetric multiplier condition in Theorem~\ref{thm:general_criterion} holds with
$\lambda_\alpha=\lambda$ for all $\alpha\in \mathcal{A}$.

By Lemma~\ref{lem:affine-moment},
\[
  \binom{m+s-1}{s-1}\sum_{\alpha\in \mathcal{A}}\alpha^m=0,
  \qquad 0\le m\le rs-s-1.
\]
Therefore the moment criterion in Theorem~\ref{thm:general_criterion} is satisfied, and hence $C$ is reverse self-dual.
\end{proof}

\begin{example}
   Let $\mathbb{F}_{25}=\mathbb{F}_{5}(\theta)$, where $\theta^2+2=0$. We consider an HRS code over $\mathbb{F}_{25}$ with multiplicity $s = 4$, dimension $k = 10$, and $r = 5$ evaluation points. Let $U =\mathbb{F}_5 = \{0, 1, 2, 3, 4\} \leq \mathbb{F}_{25}$. We define the evaluation set as the affine coset $\mathcal{A} = \theta + U = \{ \alpha_b := \theta + b \mid b \in U \}=\{\theta, \theta+1, \theta+2, \theta+3, \theta+4\}$. 
    Choose $\lambda=2$ and take the multiplier matrix $V$ to be:
    $$V = \begin{bmatrix}1 & 2 & 3 & 4 & 1 \\2 & 1 & 4 & 3 & 2 \\1 & 2 & 3 & 4 & 1 \\2 & 1 & 4 & 3 & 2\end{bmatrix}.$$
    The HRS code $C = \operatorname{HRS}_{10}(\mathcal{A}, V, 4, 5)$ is  generated by the following generator matrices $M(f)$ for $f(x) \in \{1, x,x^2,\cdots,x^9\}$:
$$M(1) = \begin{bmatrix}
1 & 2 & 3 & 4 & 1 \\
0 & 0 & 0 & 0 & 0 \\
0 & 0 & 0 & 0 & 0 \\
0 & 0 & 0 & 0 & 0
\end{bmatrix}, \quad M(x) = \begin{bmatrix}
\theta & 2\theta+2 & 3\theta+1 & 4\theta+2 & \theta+4 \\
2 & 1 & 4 & 3 & 2 \\
0 & 0 & 0 & 0 & 0 \\
0 & 0 & 0 & 0 & 0
\end{bmatrix}.$$

$$M(x^2) = \begin{bmatrix}
3 & 4\theta+3 & 2\theta+1 & 4\theta+3 & 3\theta+4 \\
4\theta & 2\theta+2 & 3\theta+1 & \theta+3 & 4\theta+1 \\
1 & 2 & 3 & 4 & 1 \\
0 & 0 & 0 & 0 & 0
\end{bmatrix}, \quad
M(x^3) = \begin{bmatrix}
3\theta & 2\theta & 3 & 1 & \theta \\
3 & \theta+2 & 3\theta+4 & 4\theta+3 & 3\theta+4 \\
3\theta & \theta+1 & 4\theta+3 & 2\theta+1 & 3\theta+2 \\
2 & 1 & 4 & 3 & 2
\end{bmatrix},$$

$$M(x^4) = \begin{bmatrix}
4 & 2\theta+1 & 3\theta+1 & \theta+3 & 4\theta+3 \\
4\theta & 4\theta & 1 & 3 & 3\theta \\
3 & 4\theta+3 & 2\theta+1 & 4\theta+3 & 3\theta+4 \\
3\theta & 4\theta+4 & \theta+2 & 2\theta+1 & 3\theta+2
\end{bmatrix}, \quad
M(x^5) = \begin{bmatrix}
4\theta & 3\theta+2 & 2\theta+1 & \theta+2 & 4\theta+4 \\
0 & 0 & 0 & 0 & 0 \\
0 & 0 & 0 & 0 & 0 \\
0 & 0 & 0 & 0 & 0
\end{bmatrix},$$

$$M(x^6) = \begin{bmatrix}
2 & 1 & 3 & 4 & 3 \\
3\theta & 4\theta+1 & \theta+3 & 2\theta+4 & 3\theta+3 \\
0 & 0 & 0 & 0 & 0 \\
0 & 0 & 0 & 0 & 0
\end{bmatrix}, \quad
M(x^7) = \begin{bmatrix}
2\theta & \theta+1 & 3\theta+1 & 4\theta+2 & 3\theta+2 \\
3 & 1 & 3 & 1 & 2 \\
4\theta & 3\theta+2 & 2\theta+1 & \theta+2 & 4\theta+4 \\
0 & 0 & 0 & 0 & 0
\end{bmatrix},$$

$$M(x^8) = \begin{bmatrix}
1 & 2\theta+4 & 2\theta+1 & 4\theta+3 & 4\theta+2 \\
2\theta & 4\theta+4 & 2\theta+4 & 4\theta+2 & 3\theta+2 \\
1 & 3 & 4 & 2 & 4 \\
3\theta & 4\theta+1 & \theta+3 & 2\theta+4 & 3\theta+3
\end{bmatrix}, \quad
M(x^9) = \begin{bmatrix}
\theta & \theta & 3 & 1 & 3\theta \\
3 & 4\theta+3 & 4\theta+2 & 2\theta+4 & 2\theta+1 \\
2\theta & \theta+1 & 3\theta+1 & 4\theta+2 & 3\theta+2 \\
1 & 2 & 1 & 2 & 4
\end{bmatrix}.$$
Denote $G_a = M(x^a)$ for $0 \le a \le 9$, then $C = \operatorname{span}_{\mathbb{F}_{25}}\{G_0, \dots, G_9\}$ and $\operatorname{Rev}(C) = \operatorname{span}_{\mathbb{F}_{25}}\{\operatorname{Rev}(G_0), \dots, \operatorname{Rev}(G_9)\}$.
A direct computation shows that
$\langle G_a, \operatorname{Rev}(G_b) \rangle = 0 \quad \text{for all } 0 \le a, b \le 9$.
Thus $\operatorname{Rev}(C) \subseteq C^\perp$. Then $C^\perp = \operatorname{Rev}(C)$ since $\dim \operatorname{Rev}(C) = \dim C = \dim C^\perp=10 $, which confirms that $C = C^{\perp_{\rm rev}}$ in agreement with Theorem~\ref{thm:affine-additive-construction}.
\end{example}

We next give multiplicative constructions for general multiplicity.
Let $q=p^m$, and let $H\subseteq\mathbb{F}_q^*$ be the multiplicative subgroup of
order $r_0$.  Assume that $r_0\equiv 1 \pmod p$. Let $Q=p^v$ with $v\mid m$, and assume that $Q-1\mid \frac{q-1}{r_0}.$ Equivalently, $\mathbb{F}_Q^*\subseteq (\mathbb{F}_q^*)^{r_0}:=\{x^{r_0}: x \in\mathbb{F}^*_q\}$. Let $W\subseteq\mathbb{F}_Q$ be an $\mathbb{F}_p$-linear subspace of dimension $u$, where
$0\le u<v$.  Choose $\xi \in\mathbb{F}_Q\backslash W $, and set
$  D=\xi+W.$
Then $D\subseteq \mathbb{F}_Q^*\subseteq (\mathbb{F}_q^*)^{r_0}$.
Define
$ \mathcal{A}=\{\alpha\in\mathbb{F}_q^*: \alpha^{r_0}\in D\}=\bigsqcup_{d\in D}\{\alpha\in\mathbb{F}_q^*: \alpha^{r_0}=d\}.$
For each $d\in D$, choose $\alpha_d\in\mathbb F_q^*$ such that $\alpha_d^{r_0}=d$. Then $\{\alpha\in\mathbb F_q^*: \alpha^{r_0}=d\}=\alpha_dH$ and
\[
        \mathcal A=\bigsqcup_{d\in D}\alpha_dH.
\]
That is, $\mathcal A$ is a disjoint union of multiplicative cosets of $H$, and the total size is $r=|\mathcal A|= p^ur_0$.

\begin{theorem}[A pullback union of multiplicative cosets for general multiplicity]\label{mul--pullback-general}
Let $2\le s\le p$ and $\mathcal{A}$ be defined as above, and assume that $sr=sp^ur_0$ is even.  Let $\eta\in\mathbb F_q^*$, and choose nonzero multipliers
$v_{i,\alpha}\in\mathbb F_q^*$ satisfying
\[
        v_{i,\alpha}v_{s-i+1,\alpha}
        =
        \eta\alpha^{s-r_0s},
        \qquad 1\le i\le s,
        \quad \alpha\in\mathcal A.
\]
When $s$ is odd, this condition includes the requirement that the right-hand side is a square in $\mathbb F_q^*$ for every $\alpha\in\mathcal A$; when $s$ is even, such pairwise products can always be realized once the right-hand side is prescribed.
Then the HRS code $C=\operatorname{HRS}_{\frac{sr}{2}}(\mathcal A,V,s,r)$ is reverse self-dual.
\end{theorem}

\begin{proof}
 Set $\lambda_\alpha:=v_{i,\alpha}v_{s-i+1,\alpha}
        =\eta\alpha^{s-r_0s}$.
 Therefore, by the reverse self-duality criterion in Theorem~\ref{thm:general_criterion}, it is enough to prove
\[
        \binom{m+s-1}{s-1}
        \sum_{\alpha\in\mathcal A}\lambda_\alpha\alpha^m=0,
        \qquad 0\le m\le sr-s-1.
\]
Since $ \mathcal A=\bigsqcup_{d\in D}\alpha_dH$, we obtain
\begin{align}
    \sum_{\alpha\in\mathcal A}\lambda_\alpha\alpha^m
        &=\eta\sum_{\alpha\in\mathcal A}\alpha^{m+s-r_0s}\nonumber\\
        &=\eta\sum_{d\in D}\sum_{\alpha \in \alpha_dH}\alpha^{m+s-r_0s}\nonumber\\
        &=\eta\sum_{d\in D}
          \sum_{\zeta \in H}(\alpha_d \zeta)^{m+s-r_0s}\nonumber\\
          &=\eta\sum_{d\in D}\frac{\alpha_d^{m+s}}{d^s}
          \sum_{\zeta \in H} \zeta^{m+s} \label{eq7}.
\end{align}
The last equality follows from $\alpha_d^{r_0}=d$ and $\zeta^{r_0}=1$. Since $H$ is cyclic of order $r_0$,
\[
        \sum_{\zeta\in H}\zeta^{m+s}
        =
        \begin{cases}
        0, & r_0\nmid m+s,\\
        r_0, & r_0\mid m+s.
        \end{cases}
\]
 Hence, if $r_0\nmid m+s$, the sum in \eqref{eq7} is zero and the desired moment identity follows. Assume that
$ m+s=hr_0$ for some integer $h\ge1$. Then, by \eqref{eq7}, we obtain
\begin{equation}\label{eq8}
    \sum_{\alpha\in\mathcal A}\lambda_\alpha\alpha^m=\eta r_0\sum_{d\in D}\frac{\alpha_d^{r_0h}}{d^s}=\eta r_0\sum_{d\in D}d^{h-s}.
\end{equation}
If $h<s$, then
 $0\le h-1<s-1<p$. Moreover, $r_0\equiv1\pmod p$, so the least base-$p$ digit of $hr_0-1$ is $h-1$. By Lucas' theorem,
\[
        \binom{m+s-1}{s-1}=\binom{hr_0-1}{s-1}
        \equiv
        \binom{h-1}{s-1}=0
        \pmod p,
\]
 and the required identity holds.

If $h\ge s$, put $ j:=h-s\ge0$. Because $s-1<p$, by Lucas' theorem and $r_0\equiv1\pmod p$, we have
\begin{equation*}\label{eq:lucas-pullback-simplified}
        \binom{m+s-1}{s-1}
        =\binom{hr_0-1}{s-1}
        \equiv \binom{h-1}{s-1}
        =\binom{j+s-1}{s-1}
        \pmod p.
\end{equation*}
Therefore, by \eqref{eq8}, we have
\begin{equation}\label{eq9}
     \binom{m+s-1}{s-1}
        \sum_{\alpha\in\mathcal A}\lambda_\alpha\alpha^m=\eta r_0 \binom{j+s-1}{s-1}\sum_{d\in D}d^{j}
\end{equation}
We next check that $j$ lies in the range where Lemma~\ref{lem:affine-moment} applies. Since
$  m\le sr-s-1=sp^ur_0-s-1,$
we have
\[
        hr_0=m+s\le sp^ur_0-1<sp^ur_0.
\]
Thus $h<sp^u$, and hence $h\le sp^u-1$. Consequently,
\[
        0\le j=h-s\le sp^u-s-1.
\]
Now apply Lemma~\ref{lem:affine-moment} to the affine additive subspace $D=\xi+W$, which has size $p^u$. Since $2\le s\le p$ and $0\le j\le sp^u-s-1$, the lemma gives
\begin{equation*}\label{eq:affine-applied-simplified}
        \binom{j+s-1}{s-1}\sum_{d\in D}d^j=0.
\end{equation*}
Then \eqref{eq9} implies that the required moment identities hold for every $0\le m\le sr-s-1$. Hence, by Theorem~\ref{thm:general_criterion}, $C$ is reverse self-dual.
\end{proof}

\begin{example}
Consider an HRS  code over $\mathbb F_9$ with multiplicity $s = 3$, dimension $k = 6$, and $r = 4$ evaluation points. Let $g$ be a primitive element of $\mathbb F_9$ and $H = \langle g^2 \rangle = \{1, g^2, 2, g^6\}$ be the multiplicative subgroup of $\mathbb{F}_9^*$ of order $r_0 = 4$. Set $Q = 3^v$ with $v=1$, meaning the subfield is chosen as the prime field $\mathbb{F}_3$. This choice satisfies the divisibility condition $Q - 1 \mid \frac{q-1}{r_0}$,
equivalently, $\mathbb{F}_3^* \subseteq (\mathbb{F}_9^*)^4$. 	Let $W=\{0\}\subseteq\mathbb F_3$, and choose
	$\xi=1\in\mathbb F_3\backslash W$. Then $D=\xi+W=\{1\}\subseteq \mathbb{F}_3^*.$  We define the evaluation set via the multiplicative pullback $\mathcal{A} = \{ \alpha \in \mathbb{F}_9^* \mid \alpha^4 \in D \} = H=\{1,g^2,2,g^6\}$. 
Choosing the scaling constant $\eta=1$, Theorem~\ref{mul--pullback-general} prescribes that the multiplier matrix $V\in  (\mathbb{F}^*_9)^{3\times 4}$ must satisfy $v_{i, \alpha} \cdot v_{4-i, \alpha} = \eta \alpha^3 / \alpha^{12} = \alpha^3$. 
We construct the multiplier matrix $V$ as follows:
$$V = \begin{bmatrix}
1 & g^2 & 2 & g^6 \\
1 & g^3 & g^2 & g \\
1 & 2 & 1 & 2
\end{bmatrix}$$
Then the HRS code $C = \operatorname{HRS}_6(\mathcal{A}, V, 3, 4)$ is generated by 
$$M(1) = \begin{bmatrix}
1 & g^2 & 2 & g^6 \\
0 & 0 & 0 & 0 \\
0 & 0 & 0 & 0
\end{bmatrix}, \quad
M(x) = \begin{bmatrix}
1 & 2 & 1 & 2 \\
1 & g^3 & g^2 & g \\
0 & 0 & 0 & 0
\end{bmatrix}, \quad
M(x^2) = \begin{bmatrix}
1 & g^6 & 2 & g^2 \\
2 & 2g^5 & g^2 & 2g^7 \\
1 & 2 & 1 & 2
\end{bmatrix}.$$

$$M(x^3) = \begin{bmatrix}
1 & 1 & 1 & 1 \\
0 & 0 & 0 & 0 \\
0 & 0 & 0 & 0
\end{bmatrix},\quad
M(x^4) = \begin{bmatrix}
1 & g^2 & 2 & g^6 \\
1 & g & 2g^2 & g^3 \\
0 & 0 & 0 & 0
\end{bmatrix}, \quad
M(x^5) = \begin{bmatrix}
1 & 2 & 1 & 2 \\
2 & 2g^3 & 2g^2 & 2g \\
1 & 2g^6 & 2 & 2g^2
\end{bmatrix}.$$
Denote $G_m = M(x^m)$ for $0 \le m \le 5$, then $C = \operatorname{span}_{\mathbb{F}_9}\{G_0, \dots, G_5\}$ and
$\operatorname{Rev}(C) = \operatorname{span}_{\mathbb{F}_9}\{\operatorname{Rev}(G_0), \dots, \operatorname{Rev}(G_5)\}$. 
 A direct computation under the Euclidean inner product yields
\[
\langle G_a, \operatorname{Rev}(G_b) \rangle = 0 \quad \text{for all } 0 \le a, b \le 5.
\]
Thus $\operatorname{Rev}(C) \subseteq C^\perp$. Then $C^\perp = \operatorname{Rev}(C)$ since $\dim C = \dim \operatorname{Rev}(C) = \dim C^\perp= 6 $.
Consequently, we have $C = C^{\perp_{\rm rev}}$, which explicitly and directly verifies the reverse self-duality asserted by Theorem~\ref{mul--pullback-general}.
\end{example}

\section{Conclusion}\label{secV}

In this paper, we studied the duality structure of Hyperderivative
Reed--Solomon codes under the NRT metric. Using Hasse derivatives and a
residue-theoretic argument, we derived an explicit component-wise
description of the Euclidean dual of an HRS code. The resulting formula
shows that the dual has a reverse-order hyperderivative evaluation
structure, together with an invertible blockwise upper-triangular
transformation determined by local expansions at the evaluation points.
In particular, for full-domain evaluation in the low-multiplicity case,
this transformation reduces to diagonal scaling, so the row-reversed dual
is again an HRS code.

Based on this duality viewpoint, we introduced and studied reverse
self-dual HRS codes. For multiplicity two, we obtained a moment criterion
for reverse self-duality and constructed several explicit families from
additive and multiplicative cosets. We further extended the criterion to
general multiplicity under a symmetric multiplier condition, and presented
corresponding constructions using additive and multiplicative structures
of finite fields. These results provide a clearer algebraic understanding
of HRS codes and may serve as a foundation for further developments such
as parity-check constructions, duality-based decoding, and Schur-product
analysis in NRT-metric spaces.

\end{document}